\documentclass[acmsmall,screen]{acmart}
\AtBeginDocument{%
  \providecommand\BibTeX{{%
    Bib\TeX}}}
    
\setcopyright{acmlicensed}
\copyrightyear{2026}
\acmYear{2026}
\acmDOI{XXXXXXX.XXXXXXX}

\acmJournal{TECS}
\acmVolume{37}
\acmNumber{4}
\acmArticle{111}
\acmMonth{8}

\usepackage{xcolor}
\usepackage{listings}
\usepackage{multirow}
\usepackage{balance}
\usepackage{xspace}
\usepackage[algoruled,longend,linesnumbered]{algorithm2e}
\BeforeBeginEnvironment{algorithm}{\setlength{\intextsep}{0pt}\setlength{\floatsep}{0pt}\setlength{\textfloatsep}{10pt}}
\usepackage{makecell}
\usepackage{forest}

\newcommand{\synprob}{{\sc OptiCC}\xspace}

\newcommand{\tool}{{\sc CiSC}\xspace}
\newcommand{\toolsol}{{\sc CiSC-Solver}\xspace}

\newcommand{\myvec}[1]{\vec{\bf #1}}
\newcommand{\hatvec}[1]{\hat{\bf #1}}

\newcommand{\angleb}[1]{\langle#1\rangle} 
 
\newcommand{\partb}[1]{{\llbracket#1\rrbracket}}

\newcommand{\dom}{\mathbb{B}}
\newcommand{\domC}{\mathbb{C}}
\newcommand{\Gone}{{\tt 1}\xspace}
\newcommand{\Gzero}{{\tt 0}\xspace}

\newcommand{\hide}[1]{}
\newcommand{\hideBFtoCirc}[1]{\hide{#1}}

\newcommand{\Gand}{{\tt and}\xspace}

\newcommand{\Gor}{{\tt or}\xspace}

\newcommand{\Gnot}{{\tt not}\xspace}

\newcommand{\Gates}{\mathbb{G}}
\newcommand{\gate}{{\tt \ell}}

\newcommand{\sem}[1]{\llbracket#1\rrbracket}

\newcommand{\ccirc}{{\bf H}}
\newcommand{\code}{{\bf C}}
\newcommand{\mmat}{{\bf M}}
\newcommand{\idmat}{{\bf I}}
\newcommand{\gmat}{{\bf G}}
\newcommand{\parti}{\mathbb{P}}

\definecolor{codegreen}{rgb}{0,0.6,0}
\definecolor{codegray}{rgb}{0.5,0.5,0.5}
\definecolor{codepurple}{rgb}{0.58,0,0.82}
\definecolor{backcolour}{rgb}{0.95,0.95,0.92}

\SetKwComment{Comment}{$\triangleright$\ }{}
\SetCommentSty{itshape}

\lstdefinestyle{mystyle}{
  backgroundcolor=\color{backcolour}, commentstyle=\color{codegreen},
  keywordstyle=\color{magenta},
  numberstyle=\tiny\color{codegray},
  stringstyle=\color{codepurple},
  basicstyle=\ttfamily\footnotesize,
  breakatwhitespace=false,         
  breaklines=true,                 
  captionpos=b,                    
  keepspaces=true,                 
  numbers=left,                    
  numbersep=5pt,                  
  showspaces=false,                
  showstringspaces=false,
  showtabs=false,                  
  tabsize=2
}

\lstdefinestyle{mycustomstyle}{
    backgroundcolor=\color{white},   
    basicstyle=\ttfamily\footnotesize,
    breaklines=true,
    frame=single,  
    numbers=none,
    keywordstyle=\color{blue}\bfseries,
    commentstyle=\color{gray},
    stringstyle=\color{red},
    morekeywords={synth-fun, declare-var, constraint, check-synth},
    framerule=0.8pt, 
    rulecolor=\color{black}  
}

\def\BibTeX{{\rm B\kern-.05em{\sc i\kern-.025em b}\kern-.08em
    T\kern-.1667em\lower.7ex\hbox{E}\kern-.125emX}}

\begin{document}
\title{Optimal Circuit Synthesis of Linear Codes for Error Detection and Correction}
%

\author{Xi Yang}
\affiliation{
  \institution{ShanghaiTech University}
  \city{Shanghai}
  \country{China}
}
\email{yangxi2023@shanghaitech.edu.cn}
\orcid{0009-0009-6648-0047}

\author{Taolue Chen}
\affiliation{
  \institution{Birkbeck, University of London}
  \city{London}
  \country{UK}
  }
\email{t.chen@bbk.ac.uk}
\orcid{0000-0002-5993-1665}

\author{Yuqi Chen}
\affiliation{
  \institution{ShanghaiTech University}
  \city{Shanghai}
  \country{China}
}
\email{chenyq@shanghaitech.edu.cn}
\orcid{0000-0003-2988-6012}

\author{Fu Song}
\authornote{Corresponding author.}
\affiliation{
  \institution{Key Laboratory of System Software (Chinese Academy of Sciences), Institute of Software, Chinese Academy of Sciences}
  \city{Beijing}
  \country{China}
  }
\email{songfu@ios.ac.cn}
  \additionalaffiliation{
  \institution{State Key Laboratory of Cryptology}
  \city{Beijing}
  \country{China}
  }
  \additionalaffiliation{
  \institution{University of Chinese Academy of Sciences}
  \city{Beijing}
  \country{China}
 }
\additionalaffiliation{
\institution{Nanjing Institute of Software Technology}
  \city{Nanjing}
  \country{China}
 }
\orcid{0000-0002-0581-2679}

\author{Chundong Wang}
\affiliation{
  \institution{ShanghaiTech University}
  \city{Shanghai}
  \country{China}
}
\email{wangchd@shanghaitech.edu.cn}
\orcid{0000-0001-9069-2650}

\author{Zhilin Wu}
\affiliation{ 
  \institution{Key Laboratory of System Software (Chinese Academy of Sciences), Institute of Software, Chinese Academy of Sciences}
  \city{Beijing}
  \country{China}
  } 
\email{wuzl@ios.ac.cn}
  \additionalaffiliation{
\institution{University of Chinese Academy of Sciences}
  \city{Beijing}
  \country{China}
  }
\orcid{0000-0003-0899-628X}

\renewcommand{\shortauthors}{Yang et al.}

\begin{abstract}
Fault injection attacks deliberately inject faults into a device via physical channels to disturb its regular execution. Adversaries can effectively deduce secrets by analyzing both the normal and faulty outputs, posing serious threats to cryptographic primitives implemented in hardware. An effective countermeasure to such attacks is via redundancy, commonly referred to as concurrent error detection schemes, where 
Binary linear codes have been used to defend against fault injection attacks. 
However, designing an optimal code circuit is often time-consuming, error-prone, and requires substantial expertise. In this paper, we formalize the optimal code circuit synthesis problem (\synprob) based on two domain-specific minimization objectives on individual inputs and parity size. We then propose a novel algorithm \tool for solving \synprob, prioritizing the minimization of individual inputs. Our approach features both correct-by-construction and secure-by-construction.
In a nutshell, \tool gradually reduces individual inputs and parity size by checking, via SMT solving, the existence of feasible Boolean functions for implementing a desired code. We further present an effective technique to lazily generate  combinations of inputs to Boolean functions, while quickly identify equivalent ones.
We implement our approach in a tool \tool, and evaluate it on practical benchmarks. Experimental results show our approach can synthesize code circuits that significantly outperform those generated by the latest state-of-the-art techniques.
\end{abstract}



\received{20 February 2007}
\received[revised]{12 March 2009}
\received[accepted]{5 June 2009}
%
%
\maketitle              
%
%
%

\section{Introduction}

Cryptographic circuits are hardware implementations of cryptographic primitives, which play an essential role in ensuring secure communication, secure authentication, privacy 
and data integrity, due to rising security risks in, e.g., 
embedded devices, 
cyber-physical systems, 
and the Internet of Things~\cite{AtzoriIM10,TYAGI202122}. %
Research has showed that fault injection (FI) attacks~\cite{BonehDL97} pose serious threats to such circuits. FI attacks deliberately disrupt their regular execution 
by introducing faults, allowing adversaries to infer secret information via an analysis of both normal and faulty outputs. Extensive studies have been done on FI attacks, with most focusing on the various physical channels used to inject faults (e.g., clock glitch~\cite{AgoyanDNRT10}, underpowering~\cite{elmaneGD08}, voltage glitch~\cite{zussa2013power}, electromagnetic pulse~\cite{DehbaouiDRT12}, or laser beam~\cite{skorobogatov2003optical}), as well as  
analysis techniques (e.g., differential fault analysis~\cite{biham1997differential},
differential fault intensity analysis~\cite{GhalatyYTS14},
fault sensitivity analysis~\cite{LiSGFTO10},
statistical fault analysis~\cite{fuhr2013fault}, 
(statistical) ineffective fault analysis~\cite{clavier2007secret,   
 dobraunig2018statistical}).

Mitigation schemes have been studied to defend against
FI attacks, typically achieved by incorporating redundancy to detect and/or correct injected faults during computation~\cite{Malkin2006ACC,KulikowskiWK08,AkdemirWKS12,AzziBCV17,AghaieMRSSS20,ShahmirzadiR020,RasoolzadehS021,MullerM24,BertoniBKMP03}. These approaches are commonly referred to as error-detecting and error-correcting schemes. \emph{Binary linear codes} (BLCs) are one of the most promising 
techniques, owing to their high performance efficiency in symmetric cryptography and a straightforward integration with other countermeasures against related side-channel attacks (e.g., masking schemes for power-based side-channel attacks)~\cite{AzziBCV17,AghaieMRSSS20,ShahmirzadiR020,RasoolzadehS021,MullerM24}. 
In a nutshell, BLCs are an error-correcting code where messages are encoded as codewords such that any linear combination of codewords is itself a codeword. They can be configured to detect/correct more faults by increasing the code's minimum (Hamming) distance (albeit at the cost of overhead). 

To implement a BLC-based fault-resistant cryptographic circuit, it is common practice for an engineer to select a concrete BLC for a given cryptographic primitive, which is then implemented as a \emph{code circuit}~\cite{AghaieMRSSS20, ShahmirzadiR020,RasoolzadehS021}. The code circuit must be correct and secure, while introducing minimal overhead in terms of, e.g., the area and latency.
As a result, this process is usually time-consuming, error-prone and expertise-demanding, especially when optimization is a concern, because it involves a large joint search space: the space of linear code (with a given message size and minimum distance) and the space of possible implementations with multiple optimization objectives (e.g., area and latency). Even worse, the optimal implementation of an optimal code (i.e., minimal parity size) does not necessarily yield an optimal solution,
due to additional domain-specific requirements. 
For instance, to satisfy the independence requirement, the input signals of some sub-circuits must be individually generated; the sum of the individual inputs should ideally be minimized~\cite{MullerM24}. As a result, conventional circuit optimization techniques in EDA tools~\cite{Tu2024} fail to generate optimal solutions for this needle-in-the-haystack problem. To our knowledge, only a greedy algorithm and a brute-force algorithm were proposed recently~\cite{MullerM24}, 
which searches for a code according to the message size and minimum distance, and synthesizes the circuit using sum-of-product (SOP) minimization~\cite{quine1952,mccluskey1956}. 
However, neither the code nor the code circuit is guaranteed to be optimal.

In this paper, we formulate the problem of {\sc Opti}mal Circuit synthesis of linear Code for error detection/correction (\synprob), with two domain-specific minimization objectives, i.e., individual inputs and parity size, which play a key role for code circuit. 
To solve this problem, we propose a novel algorithm, 
prioritizing the minimization of individual inputs.
In a nutshell, we begin with the code circuit produced by the greedy algorithm~\cite{MullerM24} and gradually reduce individual inputs and parity size by checking, via SMT solving, whether there exists feasible Boolean functions
satisfying the requirements of error detection/correction code (e.g., linearity, injectivity, minimum distance, number of individual inputs). 
As the SMT constraints enforce both the correctnest and security requirements, our approach features both correct-by-construction and secure-by-construction.

However, exhaustively exploring all possible combinations of inputs to Boolean functions according to the number of individual inputs is computationally expensive and does not scale well. To tackle this issue, 
inspired by symmetry reduction in model checking~\cite{ClarkeEJS98},
we propose an effective technique to lazily generate 
input combinations and quickly identify equivalent ones by storing 
previously explored combinations in a tree structure.
Our algorithm provides local optimality guarantees: once the Boolean functions are identified with $n_{\tt in}$ individual inputs and parity size $r$, there are no solutions that use $n_{\tt in}'<n_{\tt in}$ individual inputs 
and 
parity size $r'\leq r+1$. 
Finally, 
we synthesize the code circuit using SOP minimization, as in \cite{MullerM24}.  
As such, we obtain an optimal solution for the \synprob problem instance
which not only minimizes the number of individual inputs and parity size, but also reduces the number of gates in the resulting circuit.

We implement our algorithm in a prototype tool,
C{\sc i}rcuit Synthesizer of Linear Codes for error detection/correction (\tool),
based on the SMT solver CVC5~\cite{BarbosaBBKLMMMN22}.
We conduct extensive experiments using the recommended message sizes 
with various minimum distances to evaluate 
the effectiveness of our approach.
The experimental results show that \tool can synthesize code
circuits that significantly outperform those generated by the state-of-the-art greedy algorithm and its optimized variant via brute-force search~\cite{MullerM24}, achieving reductions in the number of individual inputs,   outputs and gates, as well as lower latency. 

In summary, our main contributions include 
\begin{itemize}
    \item A formulation of 
    \synprob, addressing the minimization of individual inputs and parity sizes, as well as correctness and security requirements;
    \item A novel algorithm for solving the \synprob problem via 
    SMT solving, featuring both correct-by-construction and secure-by-construction.
    \item A prototype tool implementing our approach, with extensive experiments confirming its effectiveness and efficiency.
\end{itemize} 

\smallskip
\noindent
  {\bf Outline.}
  Section~\ref{sec:preliminary}  briefly
  introduces circuits and linear codes.
 Section~\ref{sec:problem} formulates the \synprob problem.
 Section~\ref{sec:method} solves the \synprob problem.
 Section~\ref{sec:experiments} reports the experimental results. Finally, we discuss related work in Section~\ref{sec:relatedwork}
 and conclude the paper in Section~\ref{sec:conclusion}.


\section{Preliminaries}
\label{sec:preliminary}
 
Let $\dom$ denote the Boolean domain $\{\Gzero,\Gone\}$ (i.e., $\mathbb{F}_2$) and $[n]$ denote $\{1,2, \cdots, n\}$. 
Given a Boolean vector $\myvec{x}\in\dom^n$, $\myvec{x}_i$ denotes the $i$-th entry for $i\in [n]$.

\subsection{Circuits}
We consider three types of logic gates:  
logical-\Gand ($\wedge$), logical-\Gor ($\vee$), 
and inverter gate \Gnot ($\neg$), with fan-in 2, 2 and 1, respectively. 
Let $\Gates$ denote the set of gates $\{\neg,\wedge,\vee\}$; 
other logic gates can be easily represented. 

\begin{definition}
A (combinational logic) circuit $\ccirc$ is a tuple $(I,O,V, E,\gate)$, where
\begin{itemize}
   \item $I$ and $O$ are finite sets of input and output variables, respectively;
   \item $V$ is a finite set of vertices;
   \item $E\subseteq V\times V$ is a finite set of edges, where each $(v_1,v_2)\in E$ transmits the signal over $\dom$ from $v_1$ to $v_2$;
   \item $\gate:V\rightarrow I\cup O\cup \Gates$
   is a labeling function associating each vertex $v$ with an input variable $x\in I$, an output variable $y\in O$ or a logic gate $g\in \Gates$. Moreover, each input/output variable is associated to a unique vertex without any incoming/outgoing edges, and each vertex associated with the logic gate has the same number of incoming edges as the fan-in of the logic gate;
  \item $(V, E)$ forms a Directed Acyclic Graph (DAG). 
\end{itemize}
\end{definition}
 
A circuit $\ccirc$ represents a (multi-input, multi-output) Boolean function $\sem{\ccirc}:\dom^{|I|}\rightarrow \dom^{|O|}$ such that for any input signals $\myvec{x}\in \dom^{|I|}$, $\sem{\ccirc}(\myvec{x})$ is the output of $\ccirc$
under the input $\myvec{x}$.
 
To prevent fault propagation, the concept of \emph{independence} was proposed~\cite{AghaieMRSSS20,ShahmirzadiR020,RasoolzadehS021}
to implement cryptographic primitives that are resilient to FI attacks. It requires that the sub-circuit computing each output must be implemented independently, namely, no two sub-circuits share any logic gates. (Under this condition, a  fault injected to any single gate can only affect at most one output.)

\begin{definition}\label{def:independence}
  A circuit $\ccirc=(I,O,V, E,\ell)$ is \emph{independent} if it is a 
  composition of the sub-circuits $\{\ccirc_i\}_{i=1}^{|O|}$ such that 
\begin{itemize}
    \item for any $\myvec{x}\in \dom^{|I|}$: $\sem{\ccirc}(\myvec{x})=(\sem{\ccirc_1}(\myvec{x}),\cdots, \sem{\ccirc_{|O|}}(\myvec{x}))$; 
    \item for any $1\leq i\neq j\leq |O|$: $\ccirc_i$ and $\ccirc_j$   share no vertices other than those which are associated with input variables.
\end{itemize}   
\end{definition}
Intuitively, the single-output
Boolean functions $\sem{\ccirc_1},\cdots, \sem{\ccirc_{|O|}}$ form
the coordinate functions of the multi-output
Boolean function $\sem{\ccirc}$.

\subsection{Binary Linear Codes}  
A binary linear code (BLC~\cite{MacWilliamsS77}) is an encoding scheme that represent a message as a binary sequence, enabling the detection and correction of errors 
caused by accidental bit flips during transmission. They have been extensively used in communication and data storage to mitigate noise and prevent data corruption.
Recently, BLCs have been adopted to design and implement fault-resistant cryptographic circuits to defend against FI attacks~\cite{AghaieMRSSS20,ShahmirzadiR020,RasoolzadehS021,MullerM24}.  

\begin{definition}
A \emph{binary code} $\code_{n,k}:\dom^k\rightarrow \dom^n$ of length $n$ and rank $k$ (with $k<n$) is an injection, i.e., each message $\myvec{x}\in\dom^k$
is mapped to a unique codeword $\code_{n,k}(\myvec{x})\in\domC\subset\dom^n$, where $\domC$ forms the codeword space, i.e., the range of $\code_{n,k}$. 
$r:=n-k$ is referred to as the \emph{parity size}.

$\code_{n,k}$ is \emph{linear} if
$\domC$ is a linear subspace over $\dom$. 

$\code_{n,k}$ is \emph{systematic}
if for each message $\myvec{x}\in\dom^k$,
$\code_{n,k}(\myvec{x})$ is the concatenation
$\myvec{x}\| \myvec{x}'$ of the message $\myvec{x}$ and its  parity $\myvec{x}'$ 
derived from $\myvec{x}$.
\end{definition}
Henceforth, code refers to binary linear systematic code unless otherwise stated. 
A code $\code_{n,k}$ can alternatively be represented by
a  $k\times n$ matrix $\mmat_{k\times n}$,  
referred to as \emph{generator matrix}, which has full row rank (as $\code_{n,k}$ is injective), i.e., 
its rows form a basis for the codeword space 
$\domC$. Thus, $\code_{n,k}(\myvec{x})=\myvec{x}\cdot \mmat_{k\times n}$ for each message $\myvec{x}\in\dom^k$. 
Furthermore, for a systematic code $\code_{n,k}$, 
the generator matrix $\mmat_{k\times n}$
can be 
represented in 
the form $[\idmat_{k},\gmat_{k\times r}]$,
where $\idmat_{k}$ is the $k\times k$ identity matrix,
and $\gmat_{k\times r}$ is the \emph{parity generator matrix} (PGM) of $\code_{n,k}$. 
In other words, $\code_{n,k}(\myvec{x})= \myvec{x}\| \myvec{x}'$ where 
$\myvec{x}'=\myvec{x}\cdot \gmat_{k\times r}$.
Systematic codes enable a separation between the message and its parity.
To defend against FI attacks,
the PGM $\gmat_{k\times r}$ must be injective~\cite{BartkewitzBMMS22}, 
implying that $r\geq k$ and $n\geq 2k$.

The error detection/correction capability of $\code_{n,k}$ is determined by its minimum (Hamming) distance, defined as
\[d:= \min_{\myvec{x},\myvec{y}\in\domC.\myvec{x}\neq\myvec{y}} {\tt HW}(\myvec{x}\oplus\myvec{y}), \ \text{or equivalently,}\  d:=\min_{\myvec{x}\in \domC.\myvec{x}\neq \Gzero^n} {\tt HW}(\myvec{x})\]
where ${\tt HW}(\myvec{z})$ denotes the Hamming weight (i.e., the number of $1$-bits) of
$\myvec{z}\in\dom^n$. 
The larger the minimum distance is, the more flipped bits the code can detect/correct. 
%
The Singleton bound~\cite{Lint88} entails that $d\leq n-k+1$,  thus $r=n-k\geq d-1$.
%
Furthermore,  any linear non-systematic code admits
a systematic encoding with the same minimum distance~\cite{Blahut2003}.

When a codeword $\myvec{x}\in\domC$ is transmitted, some bits may be flipped, resulting in a faulty codeword $\hatvec{{x}}\in\dom^n$, and $\myvec{e}:=\myvec{x}\oplus\hatvec{{x}}$ is called the \emph{error vector} in which its $\Gone$-bits indicate the positions
of the flipped bits. Clearly, $\hatvec{{x}}=\myvec{x}\oplus\myvec{e}$.
\begin{proposition}[\cite{ShahmirzadiR020}]\label{prop:distancebounds}
Consider a code $\code_{n,k}$ with minimum distance $d$.
For any $\hatvec{{x}}=\myvec{x}\oplus\myvec{e}$ resulting from the  codeword $\myvec{x}\in\domC$ and error vector $\myvec{e}\in\dom^n$, the following statements hold:
\begin{itemize} 
    \item  $\code_{n,k}$ can determine whether 
   $\hatvec{x}$ is correct or faulty, if
    ${\tt HW}(\myvec{e})<d$;
   \item  $\code_{n,k}$ can recover the
   correct codeword $\myvec{x}$ from $\hatvec{{x}}$, if
    ${\tt HW}(\myvec{e})<d/2$.
\end{itemize}
\end{proposition}
%
It follows from Proposition~\ref{prop:distancebounds} that $\code_{n,k}$ with minimum distance $d$ can detect at most $d-1$ flipped bits and correct at most $(d-1)/2$ flipped bits, but 
not necessarily simultaneously.
Indeed, it can detect at most $t_d$ flipped bits and correct at most $t_c$ flipped bits simultaneously 
given $t_d+t_c< d$ and $t_d>t_c$~\cite{RasoolzadehS021}.

\begin{example}
Consider the code $\code_{8,4}$ with the following generator matrix
$\gmat_{4\times 8}$. Its minimum distance is $3$, thus is able to correct at most 
$1$ flipped bits and detect at most $2$ flipped bits.

\[\mmat_{4\times 8}=[\idmat_{4},\gmat_{4\times 4}]=
\left[
\begin{array}{cccc cccc}
                1 & 0 & 0 & 0 & 1 & 1 & 0 & 0 \\
                0 & 1 & 0 & 0 & 1 & 0 & 1 & 0 \\
                0 & 0 & 1 & 0 & 1 & 1 & 1 & 0 \\
                0 & 0 & 0 & 1 & 1 & 0 & 0 & 1 \\
            \end{array}
\right]. \]
\end{example}

\section{Formulation of \synprob Problem}
\label{sec:problem}

When designing a fault-resistant cryptographic circuit $\ccirc$ using linear codes, 
the engineer must select a message size (i.e., rank $k$) and
a minimum distance $d$ as a trade-off between resistance against fault injection attacks and implementation overhead (e.g., area and latency). 
The recommended message size is $k\in\{2,3,4\}$
for optimizing the area overhead, 
and $k=1$ for the shortest latency~\cite{MullerM24}. 
Once $k$ and $d$ are fixed,  
a code $\code_{n,k}$ that has minimum distance 
$d'\geq d$ 
should be identified, where $n\geq \max(2k,d'+k-1)$ and $r\geq \max(k,d'-1)$. 
Then $\code_{n,k}$ is implemented as a code circuit $\ccirc_c$, from which $\ccirc$ will be built to achieve optimal area and/or latency. 
The search space is further restricted to systematic codes, 
reducing to finding an injective PGM 
$\gmat_{k\times r}$ with minimum distance at least $d$.

The number of individual inputs in $\ccirc_c$ plays a
primary role  in building the corresponding 
cryptographic circuit, arguably more significant than parity size, area and/or latency~\cite{MullerM24}. 
Indeed, when $\ccirc_c$ is composed with another circuit $\ccirc'$, their composition $\ccirc_c\circ \ccirc'$ must satisfy the independence property (cf. Definition~\ref{def:independence}). 
As a result, if some input $\myvec{x}_i$ is required for computing multiple (say $t$) outputs  in $\ccirc_c$ and $\myvec{x}_i$ needs to be connected to by an output $\myvec{z}_i$ of $\ccirc'$, then $t-1$ copies of the sub-circuit computing 
$\myvec{z}_i$ must be instantiated in $\ccirc'$ (so that each individual input $\myvec{x}_i$ 
is connected to by a unique output $\myvec{z}_i$ to respect independence). This replication could lead to a substantially larger overall circuit. Similarly, the parity size $r$ also plays an 
important role, as the outputs become
inputs of the corresponding decoder (of the code).

To this end, we identify two domain-specific minimization objectives for $\ccirc_c$ that are essential for building the optimal cryptographic circuit $\ccirc$, based on which  we formulate the problem of
{\sc Opti}mal Code Circuit synthesis (\synprob).

%
For each output variable $y\in O$, let ${\tt supp}_y(\ccirc)\subseteq I$ be the set of \emph{support inputs}
involved in computing $y$ in $\ccirc$. 
Let $\#_{\tt in}\ccirc:=\sum_{y\in O} |{\tt supp}_y(\ccirc)|$, i.e., the number of \emph{individual inputs} for computing outputs. (Note that $\#_{\tt in}\ccirc$ may be well greater than the number $|I|$ of input variables.) 
Let $\#_{\tt out}\ccirc$ be the number for output variables used in $\ccirc$, i.e., $|O|$, which 
is the parity size of the code.
We consider two minimization objectives: 
number of individual inputs
and parity size, with 
priority given to the former objective, as it 
is often much larger. 





\begin{definition}
An \synprob problem is specified as a tuple
$(k,d)$ where $k$ and $d$ are the message size and minimum distance, respectively. 
\end{definition}

A \emph{solution} to an \synprob instance $(k,d)$ is a circuit $\ccirc$ such that 
\begin{itemize} 
    \item $\ccirc$ implements an injective PGM $\gmat_{k\times r}$ of a code $\code_{n,k}$,  for some parity
    size $r$ and $n=k+r$, with the minimum distance at least $d$;
    \item $\ccirc$ is independent (cf. Definition~\ref{def:independence}); 
    \item $\ccirc$ is \emph{optimal} 
    w.r.t. the two minimization objectives. 
\end{itemize}

Finding an optimal solution to an \synprob instance means 
identifying both the optimal PGM $\gmat_{k\times r}$ and 
its optimal circuit implementation. 
%
Note that we do not take circuit-specific objectives (i.e., area and latency of the code circuit) into account, 
as they are independent of other sub-circuits in
the final cryptographic
circuit. 
We note that for any (not necessarily
optimal) solution
$\ccirc$  of an \synprob instance 
$(k,d)$, we have $k\leq \#_{\tt in}\ccirc\leq k*r$,
i.e., each input must be used at least once and 
each output  has at most $k$ support inputs.


\begin{theorem}\label{thm:upper-bound}
An optimal solution of an \synprob instance $(k,d)$ can be determined in $\mathcal{O}(\frac{r^2}{k}\cdot 2^{r\cdot k+2k})$ time when the parity size is bounded by $r$, using at most $\mathcal{O}(\frac{r}{k}\cdot 2^{k})$ gates
and the longest path has at most $\mathcal{O}(k)$ gates.
\end{theorem}
 \begin{proof}  
Fix a parity size $r'$ such that $\max(k,d-1)\leq r'\leq r$. There are $2^{k\cdot r'}$ possible 
parity generator matrices $\gmat_{r',k}$, each of which
can be implemented by $r'$ coordinate functions $f_i:\dom^{k}\rightarrow\dom$ for $i\in[r']$.
Hence, there are no more than $\mathcal{O}(r\cdot 2^{k\cdot r})$ code circuits and $\mathcal{O}(2^{k\cdot r})$ coordinate functions 
to be considered.

Each coordinate function $f_i:\dom^{k}\rightarrow\dom$ can be implemented by a circuit with size (i.e., number of gates) upper-bounded by $\mathcal{O}(2^{k}/k)$ (note: $\mathcal{O}(2^{k})$ if common sub-circuits cannot be reused) and the longest path bounded by $\mathcal{O}(k)$ in length~\cite{Lupanov58,FrandsenM05}.
Thus, each code circuit has at most $\mathcal{O}(r\cdot 2^{k}/k)$ gates and the longest path is bounded by $\mathcal{O}(k)$ in length as well.
The injectivity and minimum distance of a code circuit
can be checked in 
time $\mathcal{O}(r\cdot 2^{2k}/k)$.
By enumerating all the possible code circuits and selecting an optimal one among them can be done in $\mathcal{O}(\frac{r^2}{k}\cdot 2^{r\cdot k+2k})$ time. 
\end{proof}

Obviously, exhaustive search 
is computationally expensive and
does not scale well. Below, we present a novel algorithm for solving the \synprob problem
which provides local optimality guarantees.

\section{Solving \synprob Problem}\label{sec:method}

We first recall the greedy algorithm~\cite{MullerM24} 
and present a 
procedure {\sc FindCoorFunc} for finding the best 
implementation of a PGM with given number of individual inputs and parity size.  
Based on them,  we present an overall procedure. 
\hideBFtoCirc{Finally, we describe various ways to synthesize the circuit implementations from the coordinate functions.}


\begin{algorithm}[t]
    \SetAlgoLined
    \arraycolsep=1pt
    \SetKwComment{Comment}{$\triangleright$\ }{}
    \caption{Greedy synthesis algorithm}\label{alg:greedyalg}
    \KwIn{Message size: $k$, minimum distance: $d$;}  
    \KwOut{Circuit $\ccirc$\Comment*[r]{$\ccirc$ implements a PGM $\gmat_{k\times r}$}}
    \SetKwProg{myproc}{Procedure}{}{}

    \myproc{{\sc GreedySyn}$(k,d)$}{
        $C\leftarrow \emptyset$\Comment*[r]{Truth table of the code}
        \For{$\myvec{x}=0$ {\bf to} $2^k-1$}{
             $\myvec{x}'\leftarrow 0$\;
             \While{$\exists \myvec{y}\|\myvec{y}\hspace{0.3pt}'\in C. 
             \left\{
             \begin{array}{c}
                  {\tt HW}(\myvec{x}\|\myvec{x}'\oplus \myvec{y}\|\myvec{y}\hspace{0.3pt}')<d 
                  \vee \myvec{x}'=\myvec{y}\hspace{0.4pt}' 
             \end{array}
             \right\} $}
             {
                $\myvec{x}'\leftarrow \myvec{x}'+1$\;
             }
             $C\leftarrow C\cup \{\myvec{x}\|\myvec{x}'\}$\;
        }
        
        $\ccirc\leftarrow \text{\sc GenCircuitSOP}(C)$\Comment*[r]{$\sem{\ccirc}(\myvec{x})=\myvec{x}'$ iff $(\myvec{x}\|\myvec{x}')\in C$} 
 
        \Return $\ccirc$\;  
    }
\end{algorithm}

\subsection{The Greedy Algorithm} \label{sect:greedy}

Given a message size $k$ and a minimum distance $d$,
the greedy algorithm~\cite{MullerM24} (rephrased in Algorithm~\ref{alg:greedyalg} and denoted by AGEFA$_{\tt g}$), computes a circuit
$\ccirc$ that implements a PGM $\gmat_{k\times r}$ of a 
lexicographic code $\code_{n,k}$,
without addressing the minimization objectives. It iterates each message $\myvec{x}\in\dom^k$  and computes the corresponding parity $\myvec{x}'$ in a lexicographic order. All the codewords are stored in a truth table $C$, representing
the PGM
$\gmat_{k\times r}$, 
which is transformed into a circuit through minimal sum-of-products (SOP) of each output by applying
the Quine-McCluskey algorithm~\cite{quine1952,mccluskey1956}. The code $\code_{n,k}$ generated by Algorithm~\ref{alg:greedyalg} is linear and systematic with minimum distance at least $d$, and its PGM $\gmat_{k\times r}$ is injective.

Based on the circuit
obtained from Algorithm~\ref{alg:greedyalg},
\cite{MullerM24} gradually increases
the parity size $r$, as long as the number of individual inputs decreases.
For each parity size $r$, it has to exhaustively explore and check all the possible truth tables against the security requirements, thus it is computationally expensive yet provides no optimality guarantees w.r.t. the two minimization objectives.
We denote this brute-force based variant by AGEFA$_{\tt bf}$.

\subsection{Procedure {\sc FindCoorFunc}}

Given a message size $k$, a minimum distance $d$,
the number of individual inputs $m$ and a parity size $r$,
the procedure 
{\sc FindCoorFunc} checks if there exists an $r$-tuple  of coordinate functions
$F=(f_1,\cdots,f_r)$,
implementing a PGM $\gmat_{k\times r}$ with $m$ individual inputs and minimum distance at least $d$. 
If $F$ exists, it returns $F$; otherwise $\emptyset$.

Alg.~\ref{alg:FindCoorFunc} describes {\sc FindCoorFunc}. It  
first computes the set $\parti$ of all possible partitions of 
$m$ into $r$ parts (line~\ref{alg:FindCoorFunc:parti}). 
Intuitively, a partition $\partb{n_1,\cdots,n_r}\in \parti$
indicates that the coordinate function $f_i$ (for $i\in[r]$) uses $n_i$ (from $k$ available) inputs. 
We note that the order of partitions
may matter; we consider 
lexicographical ascending and descending orders.

\begin{algorithm}[t]
    \SetAlgoLined 
    \arraycolsep=1pt
    \SetKwComment{Comment}{$\triangleright$\ }{}
    \caption{Procedure {\sc FindCoorFunc}}\label{alg:FindCoorFunc}
    \KwIn{\begin{tabular}{l}
         Message size: $k$,  minimum distance: $d$, number of individual inputs: $m$ \\
         parity size: $r$ such that $r\leq m\leq k*r$;
    \end{tabular}    }  
    \KwOut{An $r$-tuple $F$ of coordinate functions}
    \SetKwProg{myproc}{Procedure}{}{}

    \myproc{{\sc FindCoorFunc}$(k,d,m,r)$}{
        $\parti \leftarrow \{\partb{n_1,\cdots,n_r}\mid k\geq n_1\geq n_2\geq \cdots\geq n_r\geq 1, \sum_{i\in[r]} n_i=m \}$\;\label{alg:FindCoorFunc:parti}
        \ForEach(\Comment*[f]{Lexicographical ascending/descending order}){$\partb{n_1,\cdots,n_r}\in \parti$}{
            $\text{\sc Gen}\leftarrow \text{\sc BuildGenerator}(n_1,\cdots,n_r)$\;\label{alg:FindCoorFunc:gen}  
            \While(\Comment*[f]{Parallel with producer-consumer pattern}){\text{\sc Gen.hasNext}$()$}{ 
                $(X_1,\cdots,X_r)\leftarrow \text{\sc Gen.getNext}()$\Comment*[r]{$\forall i. |X_i|=n_i\wedge \bigcup_{i\in[r]} X_i=[k]$}
                $\Psi\leftarrow \text{\sc GenSpec}(X_1,\cdots,X_r,k,d)$\;
                \uIf{$\text{\sc SMT-Solver}(\Psi)={\tt SAT}$}{
                    ${\tt model} \leftarrow \text{\sc SMT-Solver.GetModel()}$\;
                    $F\leftarrow \text{\sc GenFunction}({\tt model})$\;
                    \Return $F$\;
                }
            }
        }
        \Return $\emptyset$\;  
    }
\end{algorithm}

From a partition $\partb{n_1,\cdots,n_r}\in\parti$, we build a generator
{\sc Gen} (line~\ref{alg:FindCoorFunc:gen}), which 
generates input combinations 
$(X_1,\cdots,X_r)$ such that for each $j\in [r]$, 
%
$X_j=\angleb{i_1,\cdots, i_{n_j}} \subseteq [k]$ consists of $n_j$ distinct 
indices (in an ascending order). 
For simplicity, we let $f_j(\myvec{x}_{(j)})$ denote $f_j(\myvec{x}_{i_1},\cdots,\myvec{x}_{i_{n_j}})$. 
%

For each input combination
$(X_1,\cdots,X_r)$, 
$\text{\sc GenSpec}(X_1,\cdots,X_r,k,d)$ builds an SMT formula $\Psi$ to capture  linearity,
the minimum distance $d$, and injection between inputs and outputs. Namely,  
\[\Psi:=  \forall \myvec{x},\myvec{y}\in \dom^k.  \Psi_1\wedge \Psi_2\wedge\Psi_3\]
where, 
\begin{itemize}
    \item $\Psi_1:= 
 \bigwedge_{i\in[r]} f_i(\myvec{x}_{(i)})\oplus f_i(\myvec{y}_{(i)})=f_i(\myvec{x}_{(i)}\oplus\myvec{y}_{(i)})$, capturing linearity; 
  \item $\Psi_2:=  \myvec{x}\neq 0^k\implies {\tt HW}(\myvec{x})+\sum_{i\in[r]} f_i(\myvec{x}_{(i)})\geq d$, 
 capturing minimum distance;
  \item  $\Psi_3:=  \myvec{x}\neq \myvec{y}\implies \bigvee_{i\in[r]} f_i(\myvec{x}_{(i)})\neq  f_i(\myvec{y}_{(i)})$, capturing injection.
\end{itemize}

Any solution 
$F=(f_1,\cdots,f_r)$ of
$\Psi$ implements an injective PGM $\gmat_{k\times r}$ 
with $m$ individual inputs and minimum distance at least $d$, and vice versa. We note that similar to the greedy algorithm, a code circuit can be built from the solution $F=(f_1,\cdots,f_r)$.

\begin{example}

Let us consider the input combination $(X_1, X_2, X_3) = (\langle 1,2 \rangle, \angleb{1}, \angleb{2})$ with 
message size $k=2$ and minimum distance $d = 3$.

The procedure $\text{\sc GenSpec}(X_1, X_2, X_3, 2, 3)$ produces the
following SMT formula $\Psi$:
\[
\Psi := \forall x_1,x_2, y_1,y_2 \in \mathbb{B}^2. \Psi_1 \wedge \Psi_2 \wedge \Psi_3
\]
where
\begin{itemize} 
    \item $\Psi_1 := \left(
\begin{array}{c}
     f_1(x_1,x_2)\oplus f_1(y_1,y_2)=f_1(x_1\oplus y_1,x_2\oplus x_2)  \\
    \wedge \quad f_2(x_1)\oplus f_2(y_1)=f_2(x_1\oplus y_1) \\
    \wedge \quad f_3(x_2)\oplus f_3(y_2)=f_3(x_2\oplus y_2)
\end{array}\right)$
    \item $\Psi_2 := (x_1\neq 0\wedge x_2\neq 0)  \implies x_1 + x_2 + f_1(x_1, x_2) + f_2(x_1) + f_3(x_2) \geq 3$
    \item $\Psi_3 := (x_1\neq y_1\wedge x_2\neq  y_2) \implies \big(
f_1(x_1, x_2) \neq f_1(y_1, y_2) 
\vee f_2(x_1) \neq f_2(y_1) 
\vee f_3(x_2) \neq f_3(y_2)
\big)$
\end{itemize}

By solving the above SMT formula $\Psi$, we will get a solution, whose truth tables are given below:
\begin{center}
\begin{tabular}{cc|c}
\multicolumn{3}{c}{$f_1: \mathbb{B}^2 \to \mathbb{B}$} \\
\hline
$x_1$ & $x_2$ & $f_1(x_1, x_2)$ \\
\hline
0 & 0 & 0 \\
0 & 1 & 1 \\
1 & 0 & 1 \\
1 & 1 & 0 \\
\hline
\end{tabular}
\qquad
\begin{tabular}{c|c}
\multicolumn{2}{c}{$f_2: \mathbb{B} \to \mathbb{B}$} \\
\hline
$x_1$ & $f_2(x_1)$ \\
\hline
0 & 0 \\
1 & 1 \\
\hline
\end{tabular}
\qquad
\begin{tabular}{c|c}
\multicolumn{2}{c}{$f_3: \mathbb{B} \to \mathbb{B}$} \\
\hline
$x_2$ & $f_3(x_2)$ \\
\hline
0 & 0 \\
1 & 1 \\
\hline
\end{tabular}
\end{center}

\bigskip

Their Boolean functions in SOP form are given below:
\begin{align*}
f_1(x_1, x_2) &= (\neg x_1\wedge x_2) \vee (x_1 \wedge \neg x_2), & 
f_2(x_1) &= x_1,  & 
f_3(x_2) &= x_2.
\end{align*}

This solution yields the following parity generator matrix $\gmat_{2 \times 3}$ and complete generator matrix $\mathbf{M}_{2 \times 5}$
of a code $\code_{5,2}$ with minimum distance $3$:

\begin{align*}
\gmat_{2 \times 3} &= \begin{bmatrix}
1 & 1 & 0 \\
1 & 0 & 1
\end{bmatrix}
&
\mathbf{M}_{2 \times 5} &= [\mathbf{I}_2 \mid \gmat_{2 \times 3}] = 
\begin{bmatrix}
1 & 0 & 1 & 1 & 0 \\
0 & 1 & 1 & 0 & 1
\end{bmatrix}
\end{align*}

The corresponding code circuit is shown in Figure~\ref{fig:example_circuit}.

\bigskip
\begin{figure}[h]
\centering
\includegraphics[width=0.35\textwidth]{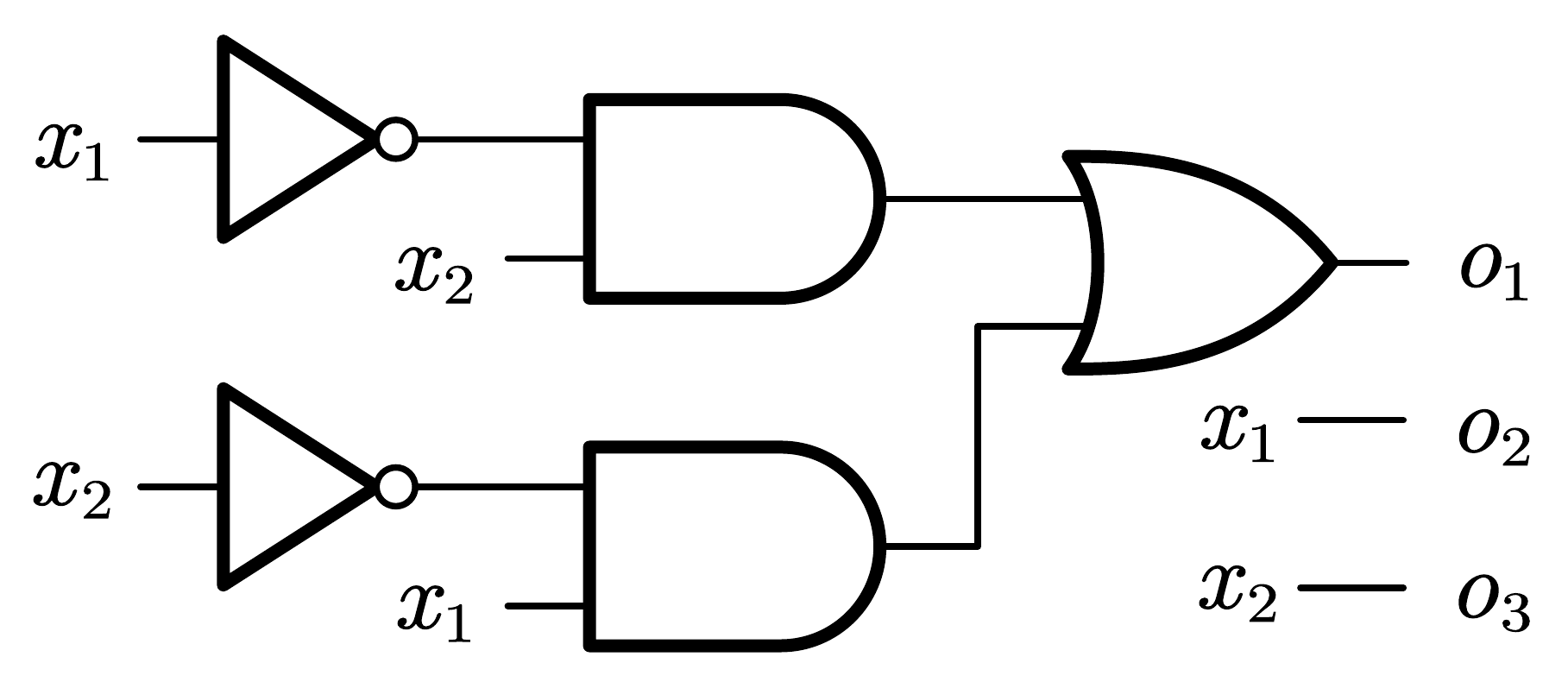}
\caption{Circuit for the code $\code_{5,2}$.}
\label{fig:example_circuit}
\end{figure}

\end{example} 
 


\begin{proposition}\label{prop:SMTbuild}
$F=(f_1,\cdots,f_r)$ satisfies
$\Psi$ iff $F$ implements an injective PGM $\gmat_{k\times r}$ 
with $m$ individual inputs and minimum distance at least $d$.
\end{proposition}

\begin{proof}
$(\Longrightarrow)$ Suppose $F=(f_1,\cdots,f_r)$ satisfies
$\Psi$. We show that $F$ implements an injective PGM $\gmat_{k\times r}$ 
with $m$ individual inputs and minimum distance at least $d$.

Since $F$ satisfies
$\Psi:=\forall \myvec{x},\myvec{y}\in \dom^k.  \Psi_1\wedge \Psi_2\wedge\Psi_3$, we have that: 

\begin{itemize}
    \item $F$ satisfies 
    $\forall \myvec{x},\myvec{y}\in \dom^k.\bigwedge_{i\in[r]} f_i(\myvec{x}_{(i)})\oplus f_i(\myvec{y}_{(i)})=f_i(\myvec{x}_{(i)}\oplus\myvec{y}_{(i)})$.
    It implies that 
    \begin{center}
        $\bigwedge_{i\in[r]} f_i(\myvec{0}_{(i)})\oplus f_i(\myvec{0}_{(i)})=f_i(\myvec{0}_{(i)}\oplus\myvec{0}_{(i)}),$
    \end{center}
    which can be simplified to $\bigwedge_{i\in[r]} f_i(\myvec{0}_{(i)})=0$. Thus, $F$ implements a PGM $\gmat_{k\times r}$.
    \item $F$ satisfies 
    $\forall \myvec{x},\myvec{y}\in \dom^k.\myvec{x}\neq 0^k\implies {\tt HW}(\myvec{x})+\sum_{i\in[r]} f_i(\myvec{x}_{(i)})\geq d$.
    Thus, the minimum distance of PGM $\gmat_{k\times r}$ is at least $d$.
    \item $F$ satisfies 
    $\forall \myvec{x},\myvec{y}\in \dom^k.\myvec{x}\neq \myvec{y}\implies \bigvee_{i\in[r]} f_i(\myvec{x}_{(i)})\neq  f_i(\myvec{y}_{(i)})$, from which
    we know that $\gmat_{k\times r}$ is injective.
\end{itemize}
From $\sum_{i\in[r]}|{\tt supp}(f_i)|=m$,
we conclude that $F$ implements an injective PGM $\gmat_{k\times r}$  with $m$ individual inputs and minimum distance at least $d$.

\smallskip
\noindent
$(\Longleftarrow)$ Suppose $F=(f_1,\cdots,f_r)$ implements an injective PGM $\gmat_{k\times r}$ 
with $m$ individual inputs and minimum distance at least $d$, we show that $F$ satisfies 
$\Psi$.

For each $i\in[r]$, let $X_i={\tt supp}(f_i)$. Since $F=(f_1,\cdots,f_r)$ uses $m$ individual inputs,
    we get that $\sum_{i\in[r]}|X_i|=m$.

\begin{itemize}
    \item Since $\gmat_{k\times r}$  is PGM, then $F$ satisfies $\forall \myvec{x},\myvec{y}\in \dom^k.\bigwedge_{i\in[r]} f_i(\myvec{x}_{(i)})\oplus f_i(\myvec{y}_{(i)})=f_i(\myvec{x}_{(i)}\oplus\myvec{y}_{(i)})$,
    i.e.,  $\forall \myvec{x},\myvec{y}\in \dom^k.\Psi_1$.
    \item Since the minimum distance of PGM $\gmat_{k\times r}$ is at least $d$, then $F$ satisfies 
    \begin{center}
            $\forall \myvec{x},\myvec{y}\in \dom^k.\myvec{x}\neq 0^k\implies {\tt HW}(\myvec{x})+\sum_{i\in[r]} f_i(\myvec{x}_{(i)})\geq d,$   i.e.,  $\forall \myvec{x},\myvec{y}\in \dom^k.\Psi_2.$
    \end{center}
    \item Since PGM $\gmat_{k\times r}$ is injective, we get that $\bigcup_{i\in [r]}X_i=[k]$ and 
 $F$ satisfies  
 \begin{center}
    $\forall \myvec{x},\myvec{y}\in \dom^k.\myvec{x}\neq \myvec{y}\implies \bigvee_{i\in[r]} f_i(\myvec{x}_{(i)})\neq  f_i(\myvec{y}_{(i)})$, i.e., $\forall \myvec{x},\myvec{y}\in \dom^k.\Psi_3$. 
 \end{center}
\end{itemize}
Therefore, $F$ satisfies
$\Psi:=\forall \myvec{x},\myvec{y}\in \dom^k.  \Psi_1\wedge \Psi_2\wedge\Psi_3$, we conclude the proof.
\end{proof}

Below, we describe the procedure for generating
the set $\parti$ of partitions and the procedure for
generating input combinations 
$(X_1,\cdots,X_r)$ for each partition $\partb{n_1,\cdots,n_r}\in \parti$.

\subsubsection{Partition Generation}
It is the partition of the integer $m$ into $r$ positive integers, each of which is no greater than $k$. 
(To ensure that no repeated partition is generated, we sort them in the descending order.) 
$\parti$ is computed by invoking $\parti(m,r,k,\partb{})$,
which is recursively defined as 
\[
\parti(m,r,k,\myvec{v}):= \left\{
\begin{array}{ll}
   \bigcup\left\{\parti(m-k',r-1,k', \myvec{v}\|k') 
\mid \begin{array}{c} 
 1\leq k'\leq k   
\end{array}\right\}, & \text{if } r>1;\\ 
\myvec{v}\|m, & \text{if } r=1.\\
\end{array}
\right.\]
where $\myvec{v}\|c$ appends $c$ at the end of $\myvec{v}$. 
Intuitively, $\parti(m,r,k,\myvec{v})$ recursively appends all
feasible integers $k'$ at the end of the partial partition $\myvec{v}$ during which $m$, $r$ and $k$ are updated accordingly. 
 
 \begin{proposition}
  $\parti(m,r,k,\partb{})$ generates the set $\parti$. \qed
 \end{proposition}
However, the size of $\parti$ cannot be defined in a closed form, instead, can be recursively computed as $\#\parti(m,r,k)$: 
\[\#\parti(m,r,k)=\left\{
\begin{array}{ll}
    \#\parti(m-k, r-1,k)+\#\parti(m,r,k-1),  & \text{if } r< m< k*r; \\
    1,  & \text{if } r= m \vee m= k*r; \\
    0,  & \text{if } r>m \vee m>k*r.
\end{array}
\right.\]
Intuitively, $\parti(m-k, r-1,k)$ is the number of partitions
where one integer is $k$ and the sum of the other
$r-1$ integers is $m-k$,
and $\#\parti(m,r,k-1)$ is the number of partitions where
no integers are $k$.

\begin{example}

To illustrate the partition generation, let us consider
the message size $k=4$, parity size $r=4$ and number of individual inputs $m=8$, i.e., $\parti(m=8, r=4, k=4, \myvec{v}=\partb{})$.
Its recursive expansion tree is shown below.
   \[ \begin{array}{ll}
        \parti(8, 4, 4, \partb{}) & \\
        \quad \xrightarrow{k'=4} & \parti(4, 3, 4, \partb{4}) \\
        & \qquad \xrightarrow{k'=2} \qquad \parti(2, 2, 2, \partb{4,2}) \\
        & \qquad\qquad \qquad\qquad \quad   \xrightarrow{k'=1}\qquad  \parti(1, 1, 1, \partb{4,2,1}) \hookrightarrow  \partb{4,2,1,1} \\[0.5em]
        \quad \xrightarrow{k'=3} & \parti(5, 3, 3, \partb{3}) \\
        & \qquad \xrightarrow{k'=3} \qquad \parti(2, 2, 3, \partb{3,3}) \\
        & \qquad\qquad\qquad\qquad\quad  \xrightarrow{k'=1} \qquad \parti(1, 1, 1, \partb{3,3,1}) \hookrightarrow \partb{3,3,1,1} \\
        & \qquad \xrightarrow{k'=2}\qquad \parti(3, 2, 2, \partb{3,2}) \\
        & \qquad\qquad\qquad\qquad \quad\xrightarrow{k'=2}\qquad \parti(1, 1, 2, \partb{3,2,2}) \hookrightarrow \partb{3,2,2,1} \\[0.5em]
        \quad \xrightarrow{k'=2} & \parti(6, 3, 2, \partb{2}) \\
        & \qquad  \xrightarrow{k'=2} \qquad \parti(4, 2, 2, \partb{2,2}) \\
        &  \qquad\qquad\qquad\qquad\quad \xrightarrow{k'=2}\qquad  \parti(2, 1, 2, \partb{2,2,2}) \hookrightarrow \partb{2,2,2,2}
    \end{array}\] 
 At each recursive level, we repeatedly select a number of inputs $k' \in [k]$ in descending order such that $r-1\leq m-k' \leq (r-1)*k'$ and then recursively invokes $\parti(m-k', r-1, k', \myvec{v})$. The four resulting partitions specify how 8 individual inputs can be allocated across 4 coordinate functions, with each function using at most 4 inputs.

Finally, the set of partitions is:
\begin{align*}
    \parti=\parti(8, 4, 4,\partb{}) = \{\partb{4,2,1,1}, \partb{3,3,1,1}, \partb{3,2,2,1}, \partb{2,2,2,2}\}
\end{align*}
which is already in the lexicographical descending order.

For the lexicographical ascending order, it is  
\begin{align*}
    \parti=\parti(8, 4, 4,\partb{}) = \{\partb{2,2,2,2},  \partb{3,2,2,1}, \partb{3,3,1,1}, \partb{4,2,1,1}\}. 
\end{align*}
\end{example}

\subsubsection{Input Combination Generation} 
Exhaustively exploring all possible input combinations is computationally expensive and
does not scale well. To tackle this issue, 
inspired by symmetry reduction in model checking~\cite{ClarkeEJS98},
we propose an effective technique to lazily generate 
input combinations and quickly identify equivalent ones by storing previously explored combinations in a tree structure.

\begin{definition}
Two input combinations $(X_1,\cdots,X_r)$
and $(X_1',\cdots,X_r')$ of the partition $\partb{n_1,\cdots,n_r}$
are \emph{equivalent}, denoted by
\begin{center}
     $(X_1,\cdots,X_r)\approx (X_1',\cdots,X_r')$,
 \end{center} 
if there are two permutations $\pi_1$ over $[r]$ and 
$\pi_2$ over $[k]$ such that   
$X_i=\pi_2(X_{\pi_1(i)}')$ for all $i\in [r]$, 
where $\pi_2$ is applied entry-wise. 
\end{definition}

Intuitively, $\pi_1$ permutes
the coordinate functions $(f_1,\cdots,f_r)$
by permuting their input indices $[r]$ while
$\pi_2$ permutes indices of inputs $[k]$.

\begin{example}
Consider the partition $\partb{2,2,1}$ for $k=3$ and $r=3$.
The input combinations $(\angleb{1,2}, \angleb{1,3}, \angleb{2})$ 
and $(\angleb{1,2}, \angleb{1,3}, \angleb{3})$ are equivalent, where 
$\pi_1$ permutes $(\angleb{1,2}, \angleb{1,3}, \angleb{3})$ to 
 $(\angleb{1,3}, \angleb{1,2}, \angleb{3})$ by exchanging
 the first and the second tuples, 
 which in turn can be permuted to $(\angleb{1,2}, \angleb{1,3}, \angleb{2})$ by $\pi_2$ (i.e., exchanging the indices $2$ and $3$ of $[k]$).
 
 In contrast, 
 $(\angleb{1,2}, \angleb{1,3}, \angleb{2})$ and $(\angleb{1,2}, \angleb{1,3}, \angleb{1})$
 are \emph{not} equivalent. 
\end{example}
 

\begin{proposition}\label{prop:symreduction}
Assume $(X_1,\cdots,X_r)\approx (X_1',\cdots,X_r')$ are generated from the same partition. 
$\Psi=\textsc{GenSpec}(X_1,\cdots,X_r,k,d)$ and $\Psi'=\textsc{GenSpec}  (X_1',\cdots,X_r', k,d)$ are equisatisfiable.
\end{proposition}

\begin{proof}
Suppose $(X_1,\cdots,X_r)\approx (X_1',\cdots,X_r')$.
Let $\pi_1$ and 
$\pi_2$ be two permutations over $[r]$ and $[k]$, respectively,
such that $X_i=\pi_2(X_{\pi_1(i)}')$ for all $i\in[r]$.
Let $\pi_1^{-1}$ and 
$\pi_2^{-1}$ be the reverse of the permutations $\pi_1$ and 
$\pi_2$, respectively. Then, for each $i\in[r]$, 
$X_i'= \pi_2^{-1}(X_{\pi_1^{-1}(i)}')$.

\smallskip
\noindent
$(\Longrightarrow)$ Suppose $F=(f_1,\cdots, f_r)$
satisfies $\Psi$.
For every $i\in[r]$, we define a function $f_i'$ such that
\[\forall \myvec{x}.\ f_i'(\myvec{x}_{(i)}) =f_{i'}(\pi_2^{-1}(\myvec{x}_{(i')})),\]
where $i':=\pi_1^{-1}(i)$ and $\pi_2^{-1}$ is applied  entry-wise.

One can verify that $F'=(f_1',\cdots, f_r')$ satisfies $\Psi'$.

\smallskip
\noindent
$(\Longrightarrow)$ Suppose $F'=(f_1',\cdots, f_r')$ that satisfies $\Psi'$.
%
For every $i\in[r]$, we define a function $f_i$ such that
\[\forall \myvec{x}.\ f_i(\myvec{x}_{(i)}) =f_{i'}'(\pi_2(\myvec{x}_{(i')})),\]
where $i':= \pi_1(i)$ and $\pi_2$ is applied entry-wise.

One can verify that $F=(f_1,\cdots, f_r)$ satisfies $\Psi$.
Thus, we conclude the proof.
\end{proof}

By Proposition~\ref{prop:symreduction}, it is sufficient to generate
one representative for each equivalence class $[(X_1,\cdots,X_r)]_\approx$.
W.l.o.g., we fix $X_1=\angleb{1,\cdots,n_1}$.
%
As checking equivalence is non-trivial, 
for efficiency consideration we introduce the concept of combination tree, allowing us to store previously explored input combinations and quickly
identify equivalent ones. 

\begin{definition}
A \emph{combination tree} $\tau$ of a partition $\partb{n_1,\cdots,n_r}$  
is pair $(Z,\delta)$, where 
\begin{itemize}
    \item $Z\subseteq [K]^{r-1}$ is a $K$-ary tree for
    $K:=\max_{i\in [r]} \binom{k}{n_i}$, namely, $\varepsilon$ is the root and if $z\cdot i\in Z$, then $z\in Z$ and $z\cdot j\in Z$ for all $j\in[i]$. Moreover,  if $|z|\leq r-2$ (i.e., $z$ is an internal node), $z$ has $\binom{k}{n_{|z|+2}}$ children;
    \item $\delta:Z\rightarrow 2^{[k]}$ is a labeling function associating each node $z\in Z$ with a tuple 
    $\delta(z)\subseteq [k]$ consisting of
    $n_{|z|+1}$ distinct indices in an ascending order such that $\delta(\varepsilon)=\angleb{1,\cdots, n_1}$ and for any pair of nodes $z\cdot i,z\cdot j\in Z$, $\delta(z\cdot i)\neq \delta(z\cdot j)$. 
\end{itemize}
\end{definition}
A node $z\in Z$ uniquely determines the root-to-$z$ path $(z_1,z_2,\cdots, z_h)$ where $z_h=z$. The path is denoted by $\myvec{z}$.
By abuse of notation, we denote
by $\delta(\myvec{z})$ the tuple $(\delta(z_1),\cdots,\delta(z_h))$,
and $\delta(\tau):=\{\delta(\myvec{z})\mid \myvec{z}\mbox{ is a root-to-path path in }\tau\}$. 
Obviously, 
$\delta(\myvec{z})$ for each root-to-leaf path $\myvec{z}$ is an input combination of the partition $\partb{n_1,\cdots,n_r}$ and $\delta(\tau)$ comprises all the input combinations of the partition $\partb{n_1,\cdots,n_r}$. 
Recall that $X_1$ is fixed to $\angleb{1,\cdots, n_1}$.

We define an equivalence relation between
two sibling nodes, allowing us to quickly identify equivalent input combinations using the tree $\tau$. 


\begin{definition} \label{def:def7}
Two nodes $z\cdot i$ and $z\cdot j$ 
are \emph{equivalent} in $\tau$, denoted by 
 \begin{center}
     $z\cdot i\simeq z\cdot j$,
 \end{center}   
if there is a bijection $\pi: \delta(z\cdot j)-\delta(z\cdot i)\rightarrow \delta(z\cdot i)-\delta(z\cdot j)$
such that for any $h\in \delta(z\cdot j)-\delta(z\cdot i)$ and any ancestor $z'\in \myvec{z}$, either $\{h,\pi(h)\} \subseteq \delta(z')$
or $\{h,\pi(h)\} \cap \delta(z')=\emptyset$.
(Note that $|\delta(z\cdot i)-\delta(z\cdot j)|=|\delta(z\cdot j)-\delta(z\cdot i)|$.)
\end{definition}
 
Intuitively, if such a  bijection $\pi: \delta(z\cdot j)-\delta(z\cdot i)\rightarrow \delta(z\cdot i)-\delta(z\cdot j)$ exists, a permutation $\pi_2$ can be derived  
which exchanges $h$ and $\pi(h)$ for $h\in \delta(z\cdot j)-\delta(z\cdot i)$ so that
$\delta(\myvec{z\cdot i})\approx \pi_2(\delta(\myvec{z\cdot j})),$
i.e., the two input combinations of the partial paths $\myvec{z\cdot i}$ and $\myvec{z\cdot j}$ are the same 
after applying $\pi_2$ to the latter combination.

\begin{example}
Let $\delta(z_1,z_2,z_2\cdot i)=(\angleb{1,2,3},\angleb{2,3},\angleb{1})$
and $\delta(z_1,z_2,z_2\cdot j)=(\angleb{1,2,3},\angleb{2,3},\angleb{2})$.
Then,  $\delta(z_2\cdot i)-\delta(z_2\cdot j)=\{1\}$
and $\delta(z_2\cdot j)-\delta(z_2\cdot i)=\{2\}$,
where $2$ appears in $\delta(z_2)=\angleb{2,3}$ but $1$ does not.
Thus $z_2\cdot i$ and $z_2\cdot j$ 
are not equivalent. 

Let $\delta(z_1,z_2,z_2\cdot i')=(\angleb{1,2,3},\angleb{2,3},\angleb{3})$. Then,
$\delta(z_2\cdot i')-\delta(z_2\cdot j)=\{3\}$
and $\delta(z_2\cdot j)-\delta(z_2\cdot i')=\{2\}$.
Since $(2,3)\subseteq \delta(z_2) \subseteq \delta(z_1)$, $z_2\cdot i'$ and $z_2\cdot j$ are equivalent. 
\end{example}

For each node $z\in Z$, we denote by $[z]_{\simeq}$ the equivalence class of $z$.
$\delta$ is extended to a set of nodes $[z]_{\simeq}$
as usual.
According to Definition~\ref{def:def7}, 
we can remove equivalent input combinations 
in $\tau$ 
%
obtaining $(Z',\delta')$, where $Z'\subseteq Z$, and, for every node $z\cdot i_1\cdots i_h\in Z$ such that 
its ancestors cannot be removed and $n_{|z|+1}>n_{|z|+2}=\cdots=n_{|z|+h+1}$,\\
\framebox{%
  \begin{minipage}{0.98\linewidth}
\centering
if $\delta(z\cdot i_1\cdots i_h)\in 
\bigcup_{j=1}^{h} \bigcup_{j'=1}^{i_j-1}
\delta([z\cdot i_1\cdots i_{j-1}\cdot j']_{\simeq})$, 
  the node
$z\cdot i_1\cdots i_h$ can be removed from $\tau$.
  \end{minipage}
}
\smallskip
  
We note, however, that
the 
tree nodes $z\cdot i$ and
the labeling function $\delta$ in the reduced tree should be adjusted accordingly to remain to be a $K$-ary tree.

\begin{example}
Let us consider the partition $\pi=\partb{2,1,1}$ with $k = 2$ and $r = 3$. Note that the number of individual inputs is $2+1+1=4$.
We generate input combinations of the form $(X_1, X_2, X_3)$
such that $|X_1| =2$, $|X_2|=|X_3|=1$ and $\bigcup_{i\in [2]} X_i=[2]$. Recall that we fix  $X_1 = \angleb{1, \cdots, |X_1|}=\angleb{1,2}$.

The full combination tree $\tau$ of $\pi$ is shown below,
which results in four input combinations:
$(\angleb{1,2},\angleb{1},\angleb{1})$, $(\angleb{1,2},\angleb{1},\angleb{2})$,
$(\angleb{1,2},\angleb{2},\angleb{1})$, and
$(\angleb{1,2},\angleb{2},\angleb{2})$.


\begin{center}
 \begin{forest}
          for tree={
            grow=south,
            circle, draw,
            minimum size=1ex,
            inner sep=1pt,
            s sep=5mm,
            l sep=2mm,
            font=\footnotesize,
            align=center
          },
          reduced/.style={draw=gray}
        [{$\varepsilon$, $X_1 = \angleb{1,2}$}, rectangle, draw=none
            [{$1$, $X_2 = \angleb{ 1}$}, rectangle, draw=none
              [{$1\cdot 1$, $X_3 = \angleb{ 1 }$ \\\\ $\rightarrow (\angleb{ 1,2}, \angleb{1}, \angleb{1})$}, rectangle, draw=none]
              [{$1\cdot 2$, $X_3 = \angleb{2}$ \\\\ $\rightarrow (\angleb{1,2}, \angleb{1}, \angleb{2})$}, rectangle, draw=none]
            ]
            [{$2$,  $X_2 = \angleb{2}$}, rectangle, draw=none, reduced
              [{$2\cdot 1$, $X_3 = \angleb{1}$ \\\\ $\rightarrow (\angleb{1,2}, \angleb{2}, \angleb{1})$}, rectangle, draw=none, reduced]
              [{$2\cdot 2$, $X_3 = \angleb{2}$ \\\\ $\rightarrow (\angleb{1,2}, \angleb{2}, \angleb{2})$}, rectangle, draw=none, reduced]
            ]
        ]
    \end{forest}
\end{center}
  
 \medskip

For the nodes $1$ and $2$, $\delta(2)-\delta(1)=\angleb{2}-\angleb{1}=\{2\}$
and $\delta(1)-\delta(2)=\angleb{1}-\angleb{2}=\{1\}$.
There is a bijection $\pi:\{2\}\rightarrow \{1\}$ such that
$\pi(2)=1$, and $\{2,\pi(2)\} \subseteq \delta(\varepsilon)$.
Thus, the nodes $1$ and $2$ are equivalent, i.e., $1\simeq 2$.
Therefore, the node $2$ can be removed in the combination tree,
as well as its descendants, leading to the reduced tree
$\tau_\simeq$,
where nodes in rectangles are removed.

The reduced tree $\tau_\simeq$ consists of only two input combinations: $(\angleb{1,2},\angleb{1},\angleb{1})$
and $(\angleb{1,2},\angleb{1},\angleb{2})$.
Indeed, they are equivalent to other two input combinations: $(\angleb{1,2},\angleb{1},\angleb{1})\approx (\angleb{1,2},\angleb{2},\angleb{2})$ and 
$(\angleb{1,2},\angleb{1},\angleb{2})\approx (\angleb{1,2},\angleb{2},\angleb{1})$.
\end{example}

\begin{proposition}\label{prop:reducedtree}
For every root-to-leaf path $\myvec{z}$ in the tree $\tau$, there exists
a root-to-leaf path $\myvec{z'}$ in the reduced tree $\tau_\simeq$ such that
$\delta(\myvec{z})\approx \delta'(\myvec{z'})$.
\end{proposition}
\begin{proof}
Let $\tau_0,\tau_1,\tau_2,\cdots,\tau_q$ be the sequence of combination trees, where $\tau_0=\tau$, $\tau_q=\tau_{\simeq}$ and for every $i\in[q]$, $\tau_i$ is obtained from $\tau_{i-1}$ by removing a node in $\tau_{i-1}$.
It suffices to show that for every $i\in [q]$, every root-to-leaf path $\myvec{\mu_{0}}$ in $\tau_{0}$, there exists a root-to-leaf path $\myvec{\mu_i}$ in $\tau_i$
such that $\delta(\myvec{\mu_0})\approx \delta(\myvec{\mu_i})$. 
The proof proceeds by applying induction on $i$.

\smallskip
\noindent{\bf Base case $i=1$}.
Consider a path $\myvec{\mu_{0}}$ that goes through the node $z\cdot i_1\cdots i_h$ which is removed from $\tau_0$, where $\delta(z\cdot i_1\cdots i_h)\in 
\delta([z\cdot i_1']_{\simeq})$ for some $i_1'<i_1$
and $n_{|z|+2}=\cdots=n_{|z|+h+1}$.
Note that the result immediately follows if the path
does not go through the removed node $z\cdot i_1\cdots i_h$.

The path $\myvec{\mu_{0}}$ can be rewritten to $\myvec{z\cdot i_1\cdots i_h \cdot z'}$, 
where ${z\cdot i_1\cdots i_h \cdot z'}$ is the leaf of the path $\myvec{\mu_{0}}$. Consequently, 
\[\delta(\myvec{\mu_{0}}):=\delta(\myvec{z})\|\delta(z\cdot i_1)\|\cdots \| \delta(z\cdot i_1\cdots i_h)\| \delta(\myvec{z'}),\]
where $\myvec{z'}$ denotes the path starting the child of the node $z\cdot i_1\cdots i_h$ to the leaf $z\cdot i_1\cdots i_h \cdot z'$, and $\delta(\myvec{z'})$ is defined as usual.

Let $\pi_1$ be the permutation over $[r]$
such that for any $j\in [r]$,
\[\pi_1(j)=\left\{\begin{array}{ll}
     j, & \text{if $j<n_{|z|+2}$ or $j>n_{|z|+h+1}$};\\ n_{|z|+h+1}, & \text{if $j=n_{|z|+2}$}; \\
     j+1,& \text{if $n_{|z|+2}< j\leq n_{|z|+h+1}$}.
\end{array}\right.\]

By permuting the components of $\delta(\myvec{\mu_{0}})$ using $\pi_1$, we get that:
\[\pi_1(\delta(\myvec{\mu_{0}})):=\delta(\myvec{z})\|
\delta(z\cdot i_1\cdots i_h)\|\delta(z\cdot i_1)\|\cdots \| \delta(z\cdot i_1\cdots i_{h-1})\| \delta(\myvec{z'}).\]

Obviously, $\delta(\myvec{\mu_{0}})\approx \pi_1(\delta(\myvec{\mu_{0}}))\in \parti$.

Since $\delta(z\cdot i_1\cdots i_h)\in 
\delta([z\cdot i_1']_{\simeq})$ with $i_1'<i_1$,
there must exist a path $\myvec{\mu_1}$ in $\tau_0$
and a permutation $\pi_2$ over $[k]$
such that (1) $\delta(\myvec{\mu_1})=\pi_2(\pi_1(\delta(\myvec{\mu_{0}})))$ where $\pi_2$ is derived 
from the bijection $\pi$ underlying the equivalence class
$[z\cdot i_1']_{\simeq})$ and applied to each input index in $\pi_1(\delta(\myvec{\mu_{0}}))$,
and moreover (2) $\myvec{\mu_1}$ goes through the node $z\cdot i_1'$. Thus $\myvec{\mu_1}$ does not go through
the node $z\cdot i_1$ and exists in $\tau_1$.

From $\delta(\myvec{\mu_{0}})\approx \pi_1(\delta(\myvec{\mu_{0}}))$ and $\delta(\myvec{\mu_1})=\pi_2(\pi_1(\delta(\myvec{\mu_{0}})))$ (i.e., $\delta(\myvec{\mu_1})\approx\pi_1(\delta(\myvec{\mu_{0}}))$), we get that
$\delta(\myvec{\mu_{0}})\approx \delta(\myvec{\mu_1})$, from which the results follows.

\smallskip
\noindent{\bf Inductive step $i>1$}.
Consider a path $\myvec{\mu_{0}}$ that goes through the node $z\cdot i_1\cdots i_h$ which is removed from $\tau_{i-1}$, where $\delta(z\cdot i_1\cdots i_h)\in 
\delta([z\cdot i_1']_{\simeq})$ for the shortest $z$ and some $i_1'<i_1$ with $n_{|z|+2}=\cdots=n_{|z|+h+1}$.
We note that the result immediately follows by applying induction hypothesis if $\myvec{\mu_{0}}$ does not go through the removed node $z\cdot i_1\cdots i_h$. 
Similarly, $\myvec{\mu_{0}}$ can be rewritten into $\myvec{z\cdot i_1\cdots i_h \cdot z'}$, 
where ${z\cdot i_1\cdots i_h \cdot z'}$ is the leaf of the path $\myvec{\mu_{0}}$. Consequently, 
\[\delta(\myvec{\mu_{0}}):=\delta(\myvec{z})\|\delta(z\cdot i_1)\|\cdots \| \delta(z\cdot i_1\cdots i_h)\| \delta(\myvec{z'}).\]

By permuting the components of $\delta(\myvec{\mu_{0}})$ using the permutation  $\pi_1$ defined the same as in the base case, we get that:
\[\pi_1(\delta(\myvec{\mu_{0}})):=\delta(\myvec{z})\|
\delta(z\cdot i_1\cdots i_h)\|\delta(z\cdot i_1)\|\cdots \| \delta(z\cdot i_1\cdots i_{h-1})\| \delta(\myvec{z'}).\]

Obviously, $\delta(\myvec{\mu_{0}})\approx \pi_1(\delta(\myvec{\mu_{0}}))\in \parti$.

Since $\delta(z\cdot i_1\cdots i_h)\in 
\delta([z\cdot i_1']_{\simeq})$ for the shortest $z$ and $i_1'<i_1$,
there must exist a path $\myvec{\mu_0'}$ in $\tau_0$
and a permutation $\pi_2$ over $[k]$
such that (1) $\delta(\myvec{\mu_0'})=\pi_2(\pi_1(\delta(\myvec{\mu_{0}})))$ where $\pi_2$ is derived 
from the bijection $\pi$ underlying the equivalence class
$[z\cdot i_1']_{\simeq})$ and applied to each input index in $\pi_1(\delta(\myvec{\mu_{0}}))$,
and moreover (2) $\myvec{\mu_0'}$ goes through the node $z\cdot i_1''$ for $i_1''=\min\{j\mid z\cdot j \in [z\cdot i_1']_{\simeq}\}$. Thus $\myvec{\mu_0'}$ does not go through the node $z\cdot i_1$
and $\delta(\myvec{\mu_{0}})\approx \delta(\myvec{\mu_{0}'})$.

By applying induction hypothesis, 
there must exist a path $\myvec{\mu_i}$ in $\tau_{i-1}$ such that (1) $\delta(\myvec{\mu_{i}})\approx \delta(\myvec{\mu_{0}'}))$ and 
(2) $\myvec{\mu_i}$
goes through the node $z\cdot i_1''$  because $z$ is the shortest one and cannot be removed,
and $i_1''=\min\{j\mid z\cdot j \in [z\cdot i_1']_{\simeq}\}\leq i_1'< i_1$. Thus, $\myvec{\mu_i}$ also exists in $\tau_{i}$.

The result follows from the fact that $\delta(\myvec{\mu_{0}})\approx \delta(\myvec{\mu_{0}'})$ and $\delta(\myvec{\mu_{i}})\approx \delta(\myvec{\mu_{0}'}))$. 
\end{proof}

Based on Proposition~\ref{prop:reducedtree}, {\sc Gen} 
generates an input combination by 
constructing 
the reduced tree $\tau_\simeq$ on-the-fly in a depth-first style. 
This is done by computing equivalence classes during the construction and checking whether a node can be added
so that a large number of equivalent input combinations can be avoided.

The order of input combinations for the same partition may impact the efficiency of \synprob solving.
Thus, we consider two strategies to permute the partition $\partb{n_1,\cdots,n_r}$ before generating input combinations: 
(1) a descending order according to component sizes 
in $\partb{n_1,\cdots,n_r}$, allowing us to fix
the more inputs early (note that 
$\parti(m,r,k,\partb{})$ already sorts the partitions in this order);
and (2) a descending order according to 
$\binom{k}{n_1},\cdots, \binom{k}{n_r}$,
allowing us to reduce input combinations early. 

\subsection{\tool Algorithm for Finding the Optimal PGM}

Alg.~\ref{alg:overallwithgreedy} presents a high-level synthesis algorithm for solving the \synprob problem $(k,d)$, which produces a circuit $\ccirc$
%
implementing a PGM $\gmat_{k\times r}$ with local optimality guarantees: 
if $\ccirc$ uses $n_{\tt in}$ individual inputs, there are no solutions $\ccirc'$ that uses $n_{\tt in}'<n_{\tt in}$  individual inputs 
and parity size $r'\leq r+1$.  

Alg.~\ref{alg:overallwithgreedy} mainly searches for $F=(f_1,\cdots,f_r)$, an $r$-tuple of coordinate functions
that implements $\gmat_{k\times r}$, from which 
$\ccirc$ can be built through minimal sum-of-products of coordinate function by applying
the Quine-McCluskey algorithm.

\begin{algorithm}[t]
    \SetAlgoLined
    \arraycolsep=1pt
    \SetKwComment{Comment}{$\triangleright$\ }{}
    \caption{\tool Algorithm}\label{alg:overallwithgreedy}
    \KwIn{Message size: $k$, minimum distance: $d$;}
    \KwOut{Circuit $\ccirc$\Comment*[r]{$\ccirc$ implements a PGM $\gmat_{k\times r}$}}
    \SetKwProg{myproc}{Procedure}{}{}
    \myproc{{\toolsol}{$(k,d)$}}{
        $\ccirc\leftarrow \text{\sc GreedySyn}(k,d)$\;\label{alg:overallwithgreedy:greedysyn}
        $n_{\tt base}\leftarrow \#_{\tt in}(\ccirc)$;
        $n_{\tt in}\leftarrow n_{\tt base}$; 
        $r_{\tt base}\leftarrow \#_{\tt out}(\ccirc)$;
        $r\leftarrow r_{\tt base}$\;
        \label{alg:overallwithgreedy:init-input-output}
        $F\leftarrow (\sem{\ccirc_1},\cdots,\sem{\ccirc_{r}})$\Comment*[r]{$\forall \myvec{x}\in \dom^k. F(\myvec{x})=\sem{\ccirc}(\myvec{x})$}\label{alg:overallwithgreedy:function}
         \label{alg:overallwithgreedy:init-input-output}
        \While(\Comment*[f]{Try to reduce $r$ with feasible $\gmat_{k\times r}$}){$r>\max(k,d-1)$} {\label{alg:overallwithgreedy:while1}
            $F'\leftarrow \text{\sc FindCoorFunc}(k,d,(r-1)*k,r-1)$\;
            \label{alg:overallwithgreedy:spec} 
            \lIf{$F'\neq \emptyset$}{$r\leftarrow r-1$}
            \lElse{{\tt break}}
        }\label{alg:overallwithgreedy:while2}
        \lIf{$r<r_{\tt base}$}{$n_{\tt in} \leftarrow n_{\tt in}+1$} \label{alg:overallwithgreedy:inincr1}
        \While(\Comment*[f]{Check and reduce $n_{\tt in}$ under the minimal $r$}){$n_{\tt in}>  \max(k,r)$}{\label{alg:overallwithgreedy:while3}
            $F'\leftarrow \text{\sc FindCoorFunc}(k,d,n_{\tt in}-1,r)$\;
            \lIf{$F'\neq \emptyset$}{$F\leftarrow F'$; $n_{\tt in}\leftarrow n_{\tt in}-1$} 
           \lElse{{\tt break}}
        }\label{alg:overallwithgreedy:while4}
        $old_{\tt in}\leftarrow n_{\tt in}$\;
        \While(\Comment*[f]{Try to reduce $n_{\tt in}$ by inc $r$}){$n_{\tt in}> \max(k,r+1)$}{\label{alg:overallwithgreedy:while5}
            \uIf{$r+1=r_{\tt base}\wedge n_{\tt in}-1=n_{\tt base}$}{
                $old_{\tt in}\leftarrow n_{\tt in}$; 
                $n_{\tt in}\leftarrow n_{\tt in}-1$; 
                \label{alg:overallwithgreedyreachbase}
            }
            $F'\leftarrow \text{\sc FindCoorFunc}(k,d,n_{\tt in}-1,r+1)$\; 
            \lIf{$F'\neq \emptyset$}{$F\leftarrow F'$; $n_{\tt in}\leftarrow n_{\tt in}-1$} 
           \lElseIf{$n_{\tt in}<old_{\tt in}$}{
                $old_{\tt in}\leftarrow n_{\tt in}$; 
                $r\leftarrow r+1$\Comment*[f]{Keep increasing $r$\label{alg:overallwithgreedyincr-r}}
                }
           \lElseIf{$r< r_{\tt base}\wedge n_{\tt in}-1=n_{\tt base}$}{$r\leftarrow r+1$\Comment*[f]{Keep increasing $r$}}
           \lElse{{\tt break}}
        }\label{alg:overallwithgreedy:while6}
         $\ccirc\leftarrow \text{\sc GenCircuitSOP}(F)$\Comment*[r]{$\forall x\in\dom^k.\sem{\ccirc}(\myvec{x})=F(\myvec{x})$} 
        \Return $\ccirc$
    }   
\end{algorithm} 

It first invokes {\sc GreedySyn} (cf. Alg.~\ref{alg:greedyalg}) to generate a code circuit $\ccirc$ (line~\ref{alg:overallwithgreedy:greedysyn}),
from which the number $n_{\tt in}$ of individual inputs and parity size $r$ (i.e., the number of outputs) are identified (line~\ref{alg:overallwithgreedy:init-input-output}) as the current smallest bases ($n_{\tt base}$
and $r_{\tt base}$).
An $r$-tuple of coordinate functions $(f_1,\cdots,f_r)$, implementing the PGM $\gmat_{k\times r}$ of $\ccirc$,
is also extracted  (line~\ref{alg:overallwithgreedy:function}).
The final solution $\ccirc'$ should be no worse than $\ccirc$, namely, either
(1) 
$\ccirc'$ uses fewer individual inputs than $\ccirc$, 
or (2) $\ccirc'$ uses no larger parity size than 
$\ccirc$ 
if $\ccirc'$ and $\ccirc$ use the same number of individual inputs. 

After that, the while-loop at lines~\ref{alg:overallwithgreedy:while1}--\ref{alg:overallwithgreedy:while2}  
gradually decreases the parity size $r$ 
until no PGM $\gmat_{k\times r}$ with minimal distance at least $d$  exists, following Propositions~\ref{prop:SMTbuild}.
It invokes the procedure $\text{\sc FindCoorFunc}(k,d,(r-1)*k,r-1)$ to check whether a PGM $\gmat_{k\times (r-1)}$ with minimal distance at least $d$ exists, without imposing any constraint on 
individual inputs, as $(r-1)*k$ is a trivial upper-bound. 
If such a PGM $\gmat_{k\times (r-1)}$ is found, $r$ is decreased by $1$. Recall that  $r\geq k$ and $r\geq d-1$, thus $r$ can only be reduced  
when $r>\max(k,d-1)$.

After this while-loop, $n_{\tt in}:=n_{\tt base}$ and $r$
is the minimal parity size for which
a PGM $\gmat_{k\times r}$ with minimal distance at least $d$ exists. There are two cases: 
\begin{itemize}
    \item  {\bf Case (i):} $r=r_{\tt base}$, i.e., $r_{\tt base}$ is already minimal, and a PGM $\gmat_{k\times r}$ with minimal distance at least $d$ is implemented by $F$ extracted from $\ccirc$ using $n_{\tt base}$ individual inputs;
    \item  {\bf Case (ii):} $r<r_{\tt base}$, i.e., $r_{\tt base}$ can be decreased to a smaller $r$ but, for any PGM $\gmat_{k\times r}$ with minimal distance at least $d$,
    an implementation 
    using $n_{\tt base}$ individual inputs may not exist.
\end{itemize}

For {\bf Case (i)}, following Propositions~\ref{prop:SMTbuild}, the second while-loop (lines~\ref{alg:overallwithgreedy:while3}--\ref{alg:overallwithgreedy:while4}) gradually decreases $n_{\tt in}$ 
by invoking the procedure $\text{\sc FindCoorFunc}(k,d,n_{\tt in}-1,r)$
until $n_{\tt in}$ becomes minimal wrt 
$r$. 
$\text{\sc FindCoorFunc}(k,d,n_{\tt in}-1,r)$ checks whether there exists an $r$-tuple $F'$ of coordinate functions,
implementing a PGM $\gmat_{k\times r}$ 
with minimal distance at least $d$ and
$n_{\tt in}-1$ individual inputs.
If $F'$ exists, $n_{\tt in}$ is decreased by $1$.
Note that such $F'$ exists only if 
each of the $k$ inputs is used at least once and each of the $r$ coordinate functions uses at least one input, i.e., $n_{\tt in}>\max(k,r)$.

For {\bf Case (ii)}, 
we first check whether there exists a PGM
$\gmat_{k\times r}$ with minimal distance at least $d$ that can be implemented by an $r$-tuple $F'$ using $n_{\tt base}$ individual inputs, which is done by  $\text{\sc FindCoorFunc}(k,d,n_{\tt in}-1,r)$.  
The rest is the the same as \textbf{Case (i)}. 

After these steps, 
there are two possibilities: 
\begin{itemize}
    \item {\bf Case (1):} $n_{\tt in}\leq n_{\tt base}$ and $r\leq r_{\tt base}$, indicating that $F'$ implements a PGM $\gmat_{k\times r}$ with minimal distance at least $d$ using 
    $n_{\tt in}$ individual inputs for 
    $r$;
 %
    \item  {\bf Case (2):} $n_{\tt in} -1= n_{\tt base}$ and $r<r_{\tt base}$, indicating that $F$ only implements $\sem{\ccirc}$ using $n_{\tt base}$ individual inputs (which may not be minimal) for the \emph{non-minimal} parity size $r_{\tt base}$. Moreover, there is no PGM $\gmat_{k\times r}$ with minimal distance at least $d$ that can be implemented 
    using $m$ individual inputs for any $m\leq n_{\tt base}$.
\end{itemize}

Depending upon {\bf Case (1)} and {\bf Case (2)}, the third while-loop at lines~\ref{alg:overallwithgreedy:while5}--\ref{alg:overallwithgreedy:while6} proceeds as follows.  

For {\bf Case (1)}, while $n_{\tt in}$ is the minimal number of individual inputs for the (minimal) parity size $r$, 
$n_{\tt in}$ may be reduced further by increasing the parity size $r$~\cite{MullerM24}.
As the number of individual inputs 
is more important, we gradually  
increase $r$ as long as $n_{\tt in}$ decreases. 
Here, the procedure $\text{\sc FindCoorFunc}(k,d,n_{\tt in}-1,r+1)$ checks if $n_{\tt in}$ can be reduced by $1$ after increasing $r$ by $1$. 
(Note that $n_{\tt in}$ may be reduced further by increasing $r$, but for efficiency consideration it is not attempted, meaning that only local optimality guarantees are given. To achieve global optimality, we could keep increasing $r$ until $r=n_{\tt in}$, when each coordinate function uses one input.)

For {\bf Case (2)}, the while-loop
first identifies a parity size $r'$
starting from $r+1$ such that 
a PGM $\gmat_{k\times r'}$ with minimal distance at least $d$ exists and can be implemented by an $r'$-tuple $F'$ of coordinate functions using $n_{\tt base}$ individual inputs. If such an $r'< r_{\tt base}$ can be found, it 
proceeds the same as in {\bf Case (1)}.

If no $r'< r_{\tt base}$ is found, we conclude that there is no PGM $\gmat_{k\times r'}$ with minimal distance at least $d$ that can be implemented by an $r'$-tuple $F'$ of coordinate functions using $m$ individual inputs for any $m\leq n_{\tt base}$ and $r'< r_{\tt base}$.
At this moment, $F$ implements $\sem{\ccirc}$ using $n_{\tt base}$ individual inputs (but may not be minimal) for the \emph{non-minimal} parity size $r_{\tt base}$. Thus, we 
decrease $n_{\tt base}$ 
for the parity size $r_{\tt base}$ (line~\ref{alg:overallwithgreedyreachbase}). 


\begin{theorem}\label{thm:alg3} 
The following two statements hold:
\begin{itemize}
    \item  {\bf Termination.} Alg.~\ref{alg:overallwithgreedy} always terminates. \item {\bf Local optimality.} Suppose $F=(f_1,\cdots,f_r)=\text{\toolsol}(k,d)$
  and $m=\sum_{i\in[r]} |{\tt supp}(f_i)|$.
For any $r'\leq r+1$ and $m'<m$, 
    the \synprob problem $(k,d)$
    has no 
    solutions $\ccirc$ that uses
    $m'$ individual inputs
    and $r'$ parity size. 
\end{itemize}
\end{theorem}

 \begin{proof}
We first prove the following proposition.

\begin{proposition}\label{prop:cutpoint}
(1) If an \synprob instance $(k,d)$ has no solutions that
use $n_{\tt in}$ individual inputs and $r$ parity size where $n_{\tt in}\leq k*r$,
then it has no solutions that use $n_{\tt in}'=n_{\tt in}-1$ individual inputs with the same parity size $r$.
(2) If it has a solution with parity size $r$, then it also has
a solution with parity size $r+1$ (without constraining on individual inputs). 
\end{proposition}

\noindent
{\bf Proof of Proposition~\ref{prop:cutpoint}.}
Suppose an \synprob instance $(k,d)$ has no solutions that
use $n_{\tt in}$ individual inputs and $r$ parity size where $n_{\tt in}\leq k*r$, but it has a solution $\ccirc$ that uses $n_{\tt in}'=n_{\tt in}+1$ individual inputs with the same parity size $r$.

Since $n_{\tt in}'<n_{\tt in} \leq k*r$, there must exist
an output $y$ in $\ccirc$ such that $|{\tt supp}_y(\ccirc)|< k$.
Let $\ccirc_y$ be the sub-circuit of $\ccirc$ for computing
$y$. Then, we can add an additional input $x$ into $\ccirc_y$, leading to a new sub-circuit  $\ccirc_y'$, and moreover
$\ccirc_y'(\myvec{x}):=\ccirc_y'(\myvec{x})\oplus x\oplus x$.

By replacing the sub-circuit  $\ccirc_y$ in $\ccirc$
with the sub-circuit $\ccirc_y'$, we obtain a new circuit $\ccirc_{+1}$ that uses one more individual input than $\ccirc$.
$\ccirc_{+1}$ implements the same injective PGM as $\ccirc$,
and is independent, thus is a new solution of the  \synprob instance $(k,d)$ using 
the same parity size $r$ as $\ccirc$.

\smallskip
Now, we prove that if it has a solution with parity size $r$, then it also has a solution with parity size $r+1$ (without constraining on individual inputs).

Let $\ccirc$ be a solution using parity size $r$.
Then, we can add a sub-circuit $\ccirc_{r+1}$ into $\ccirc'$,  leading to a new circuit $\ccirc'$, where $\ccirc_{r+1}$ always outputs $0$, namely, all the parities of $\ccirc$ are appended by $0$. Obviously, $\ccirc'$ is always a solution of the problem. 

\medskip
\noindent
{\bf Proof of Termination.}
The first while-loop either reduces $r$ by $1$ or breaks during each iteration, and it can iterate
at most $r_{\tt base}-\max(k,d-1)$ times. 
The second while-loop either reduces $n_{\tt in}$ by $1$ or breaks during each iteration, and it can iterate
at most $n_{\tt base}+1-\max(k,r)$ times. 
The third while-loop either reduces $n_{\tt in}$ by $1$ or increases $r$ by $1$ or breaks during each iteration, thus
the loop condition $n_{\tt in}> \max(k,r+1)$ will not hold eventually. Thus, Alg.~\ref{alg:overallwithgreedy} always terminates.

\medskip
\noindent{\bf Proof of local optimality.}
   %
Suppose there is a solution $\ccirc'$ that has $m'$ number of individual inputs and $r'$ parity size for some $m'<m$ and $r'\leq r+1$.

The proof proceeds by distinguishing if
$r'\leq r$ or $r'=r+1$.

 If $r'\leq r$,  by Proposition~\ref{prop:cutpoint}(1),
an $r'$-tuple of coordinate functions $F'=(\sem{\ccirc_1'},\cdots,\sem{\ccirc_{r'}'})$ must be found in either the second or the third while-loop
and thus $F$ should not be returned in Alg.~\ref{alg:overallwithgreedy}.
By Proposition~\ref{prop:cutpoint}(2), 
the first while-loop identifies the minimal
the parity size $r$ w.r.t. the minimum distance $d$.
Observe that Alg.~\ref{alg:overallwithgreedy}
increases the parity size $r$ at most one from
the minimal parity size identified in the first while-loop.
If $r'=r$, then $F$ has been found before finding $F'$ and Alg.~\ref{alg:overallwithgreedy} will try to find another 
one using less number of individual inputs in either the second or the third while-loop, eventually 
finds $F'$. If $r'<r$, then $F'$ should be found before finding $F$. 
Whenever $F'$ has been found, the number $n_{\tt in}$ of individual inputs will not increase any more (i.e., either keep the same or decrease), thus $F$ should not be returned in Alg.~\ref{alg:overallwithgreedy}.

   If $r'=r+1$,  the third while-loop would find better one $F'$, thus $F$ should not be returned.
\end{proof}

\section{Evaluation}\label{sec:experiments}

We implement our algorithm as a prototype tool, 
based on the SMT solver CVC5, and integrate it with Yosys~\cite{Yosys} for compiling and optimizing the generated Verilog  RTL modules into Verilog netlist circuits. Particularly, input combinations are processed in parallel using the producer-consumer pattern with OpenMP, a widely used API for shared-memory parallel programming\footnote{\url{https://www.openmp.org}.}. 

Belowe, we thoroughly evaluate the efficacy of 
\tool  in generating optimal code circuits 
by varying
the message size
$k$ from $1$ to $6$ and the minimum distance
$d$ from $2$ to $5$. (Recall that $1\le k \le 4$ is recommended~\cite{MullerM24}.)
We first evaluate the design choices of \tool for partition generation and input combination generation,
then compare \tool with AGEFA$_{\tt g}$ 
and AGEFA$_{\tt bf}$~\cite{MullerM24} (cf. Section~\ref{sect:greedy}).


All experiments were conducted on a machine
with an Intel(R) Xeon(R) Platinum 8375C CPU (32 cores, 2.90 GHz), 755 GB RAM and Ubuntu 20.04.6 LTS.
(755 GB RAM is the memory for evaluation instead of consumption.) 
Both AGEFA$_{\tt bf}$ and \tool uses 16 threads for finding
coordinate functions with 48-hour timeout, others use single-thread (Wall time for multi-thread and CPU time for single-thread). When timeout occurs, we use the current-best coordinate functions produced by \tool and AGEFA$_{\tt bf}$
for building circuits.

\subsection{Evaluation of Partition Generation}
 
\begin{table}[t]
    \centering\setlength{\tabcolsep}{7pt} 
    \renewcommand{\arraystretch}{1.1}
    \caption{Results of \tool w.r.t. two partition orders.}
    \vspace{-2mm}
    \label{tab:partitionorders}
  \scalebox{0.9}{  
  \begin{tabular}{|c||r|r|r|r|r|| r|r|r|r|r|}
    \hline
    \multirow{2}*{($k,d$)} &  
    \multicolumn{5}{c||}{Lexicographical ascending order} & \multicolumn{5}{c|}{Lexicographical descending order} \\ \cline{2-11}
    & \#Part & \#Comb & Area & Len & Time
    & \#Part & \#Comb & Area & Len & Time \\ \hline
        (1,2) & 0 & 0 & 0 & 1 & 0.001 & 0 & 0 & 0 & 1 & 0.001 \\ \hline
        (1,3) & 0 & 0 & 0 & 1 & 0.001 & 0 & 0 & 0 & 1 & 0.001 \\ \hline
        (1,4) & 0 & 0 & 0 & 1 & 0.001 & 0 & 0 & 0 & 1 & 0.001 \\ \hline
        (1,5) & 0 & 0 & 0 & 1 & 0.001 & 0 & 0 & 0 & 1 & 0.001 \\ \hline
        (2,2) & 0 & 0 & 0 & 1 & 0.001 & 0 & 0 & 0 & 1 & 0.001 \\ \hline
        (2,3) & 1 & 2 & 5 & 4 & 0.06 & 1 & 2 & 5 & 4 & 0.06 \\ \hline
        (2,4) & 2 & 7 & 10 & 4 & {\bf 0.09} & 2 & 7 & 10 & 4 & 0.10 \\ \hline
        (2,5) & 2 & 11 & 10 & 4 & 0.14 & 2 & 11 & 10 & 4 & 0.14 \\ \hline
        (3,2) & 0 & 0 & 0 & 1 & 0.001 & 0 & 0 & 0 & 1 & 0.001 \\ \hline
        (3,3) & 6 & 27 & 10 & 4 & 0.34 & 6 & 27 & 10 & 4 & {\bf 0.32} \\ \hline
        (3,4) & 5 & 40 & 32 & 6 & 0.30 & 5 & 40 & 32 & 6 & {\bf 0.28} \\ \hline
        (3,5) & 5 & 197 & 32 & 6 & 4.37 & 5 & 197 & 32 & 6 & 4.37 \\ \hline
        (4,2) & 0 & 0 & 0 & 1 & 0.001 & 0 & 0 & 0 & 1 & 0.001 \\ \hline
        (4,3) & 10 & 164 & \textbf{15} & \textbf{4} & \textbf{2.81} & 10 & 164 & 47 & 12 & 2.84 \\ \hline
        (4,4) & 10 & 215 & 68 & 6 & 1.56 & 10 & 215 & 68 & 6 & 1.59 \\ \hline
        (4,5) & 17 & 5898 & 59 & 6 & \textbf{87.92} & 17 & 5898 & 59 & 6 & 91.59 \\ \hline
        (5,2) & 0 & 0 & 0 & 1 & 0.003 & 0 & 0 & 0 & 1 & 0.003 \\ \hline
        (5,3) & 16 & 1120 & \textbf{20} & \textbf{4} & \textbf{29.48} & 16 & 1120 & 119 & 21 & 31.75 \\ \hline
        (5,4) & 24 & 6274 & 128 & 12 & 141.41 & 24 & 6274 & 128 & 12 & \textbf{130.5} \\ \hline
        (5,5) & 45 & 511479 & 113 & 12 & \textbf{8848.96} & 45 & 511479 & 113 & 12 & 9469.83 \\ \hline
        (6,2) & 0 & 0 & 0 & 1 & 0.001 & 0 & 0 & 0 & 1 & 0.001 \\ \hline
        (6,3) & 25 & 8967 & \textbf{25} & \textbf{4} & 467.62 & 25 & \textbf{8137} & 287 & 38 & \textbf{386.09} \\ \hline
        (6,4) & 51 & 252065 & 272 & 21 & 14878.6 & 51 & 252065 & 272 & 21 & \textbf{14134.58} \\ \hline
        (6,5) & 4 & 606240 &136 &6 & \textbf{Timeout} & \multicolumn{4}{c|}{N/A} & \textbf{Timeout} \\ \hline
    \end{tabular}}
\end{table}

We compare the effectiveness of two partition orders in \tool: 
lexicographical ascending order vs. lexicographical descending order,
where the equivalence-based reduction is enabled and 
input combinations are generated according to component sizes.

The results are reported in Table~\ref{tab:partitionorders},
including the number of partitions and the number of input combinations processed until finding an optimal solution,
the number of gates (Area) and the length of the longest path in the code circuit (Len), and execution time in seconds.
We note that for $(k,r)=(6,5)$, \tool runs out of time (48 hours)
with lexicographical ascending order, we report the results when
the better code circuit is found.

We can observe that the number of partitions and the number of input combinations processed by \tool using two different partition order
are almost the same, except for the case $(k,d)=(6,3)$. 
At first glance, it is strange, thus we checked the coordinate functions and indeed they are often different, evidenced by different circuit area and latency (i.e., Len).  In terms of execution time,
they are also almost comparable. Finally, we consider that  
the lexicographical ascending order is better than the lexicographical descending order, because the inputs are more balanced across coordinate functions and thus the code circuits
use fewer gates with lower latency.
 

\begin{table}[t]
    \centering\setlength{\tabcolsep}{12pt} 
    \renewcommand{\arraystretch}{1.1}
    \caption{Results of \tool w.r.t. two partition orders.}
    \vspace{-2mm}
    \label{tab:inputcombinationorders}
  \scalebox{0.9}{  
  \begin{tabular}{|c||r|r|r|r||r|r|}
    \hline
    \multirow{3}*{($k,d$)} &  
    \multicolumn{4}{c||}{With reduction} & 
    \multicolumn{2}{c|}{Without reduction} \\ \cline{2-7}
    & \multicolumn{2}{c|}{Component size} & \multicolumn{2}{c||}{Comb. numbers} & \multicolumn{2}{c|}{Component size}  \\ \cline{2-7}
    &\#Comb & Time & \#Comb & Time &\#Comb & Time  \\ \hline
        (1,2) & 0 & 0.001 & 0 & 0.001 & 0 & 0.001 \\ \hline
        (1,3) & 0 & 0.001 & 0 & 0.001 & 0 & 0.001 \\ \hline
        (1,4) & 0 & 0.001 & 0 & 0.001 & 0 & 0.001 \\ \hline
        (1,5) & 0 & 0.001 & 0 & 0.001 & 0 & 0.001 \\ \hline
        (2,2) & 0 & 0.001 & 0 & 0.001 & 0 & 0.001 \\ \hline
        (2,3) & 2 & 0.06 & 2 & 0.06 & 6 & 0.07 \\ \hline
        (2,4) & 7 & 0.09 & 7 & 0.09 & 38 & 0.13 \\ \hline
        (2,5) & 11 & 0.14 & 11 & 0.14 & 158 & 0.4 \\ \hline
        (3,2) & 0 & 0.001 & 0 & 0.001 & 0 & 0.001 \\ \hline
        (3,3) & 27 & \textbf{0.34} & 27 & 0.37 & 300 & 0.88 \\ \hline
        (3,4) & 40 & \textbf{0.3} & 40 & 0.33 & 426 & 1.75 \\ \hline
        (3,5) & 197 & 4.37 & 197 & \textbf{4.21} & 11787 & 103.51 \\ \hline
        (4,2) & 0 & 0.001 & 0 & 0.001 & 0 & 0.001 \\ \hline
        (4,3) & \textbf{164} & 2.81 & 168 & 2.81 & 7368 & 56.34 \\ \hline
        (4,4) & \textbf{215} & 1.56 & 243 & \textbf{1.48} & 9960 & 84.12 \\ \hline
        (4,5) & \textbf{5898} & \textbf{87.92} & 7990 & 120.65 & 2413260 & 33517.52 \\ \hline
        (5,2) & 0 & 0.003 & 0 & \textbf{0.002} & 0 & 0.007 \\ \hline
        (5,3) & \textbf{1120} & \textbf{29.48} & 1153 & 30.6 & 470094 & 6503.46 \\ \hline
        (5,4) & \textbf{6274} & \textbf{141.41} & 6850 & 151.44 & N/A & \textbf{Timeout} \\ \hline
        (5,5) & \textbf{511479} & \textbf{8848.96} & 586722 & 9877.27 & N/A & \textbf{Timeout} \\ \hline
        (6,2) & 0 & 0.001 & 0 & 0.001 & 0 & 0.001 \\ \hline
        (6,3) & 8967 & 467.62 & \textbf{8136} & \textbf{379.69} & N/A & \textbf{Timeout} \\ \hline
        (6,4) & \textbf{252065} & \textbf{14878.6} & 311209 & 18952.01 & N/A & \textbf{Timeout} \\ \hline
        (6,5) & N/A & \textbf{Timeout} & N/A & \textbf{Timeout} & N/A & \textbf{Timeout} \\ \hline
    \end{tabular}}
\end{table}

\subsection{Evaluation of Input Combination Generation}

We evaluate the effectiveness of our equivalence-based reduction
for input combination generation, as well as two orders
of input combinations, where partitions are produced in the lexicographical ascending order.

The results are reported in Table~\ref{tab:inputcombinationorders},
including the number of input combinations processed until finding an optimal solution and execution time in seconds.

First of all, we can observe that a large number of equivalent input combinations are identified by our equivalence-based reduction so that only fewer input 
combinations are checked via SMT solving. Thus,
our equivalence-based reduction significantly reduces execution time. Between two orders of input combinations, we can observe
that the descending order according to component sizes often performs better than that according to combinational numbers,
allowing us to find an optimal solution using fewer input combinations.



\begin{table}[t]
    \centering\setlength{\tabcolsep}{4pt} 
    \renewcommand{\arraystretch}{1.1}
    \caption{Results of AGEFA$_{\tt g}$, AGEFA$_{\tt bf}$, and
    \tool for generating code circuits.}
    \vspace{-2mm}
    \label{tab:code-circuit}
  \scalebox{0.9}{  \begin{tabular}{|c||r|c|r|r|r||r|c|r|r|r||r|c|r|r|r|}
    \hline
   \multirow{2}*{($k,d$)} & 
   \multicolumn{5}{c||}{AGEFA$_{\tt g}$} & 
   \multicolumn{5}{c||}{AGEFA$_{\tt bf}$} &
   \multicolumn{5}{c|}{\tool}\\   \cline{2-16} 
    & \#In & r & Area & Len & Time &
      \#In & r & Area & Len & Time  &
     \#In & r & Area & Len & Time \\ \hline\hline
        (1,2) & 1 & 1 & 0 & 1 & 0.001  & 1 & 1 & 0 & 1 & 0.003  & 1 & 1 & 0 & 1 & 0.002  \\ \hline
        (1,3) & 2 & 2 & 0 & 1 & 0.001  & 2 & 2 & 0 & 1 & 0.003  & 2 & 2 & 0 & 1 & 0.002  \\ \hline
        (1,4) & 3 & 3 & 0 & 1 & 0.002  & 3 & 3 & 0 & 1 & 0.004  & 3 & 3 & 0 & 1 & 0.002  \\ \hline
        (1,5) & 4 & 4 & 0 & 1 & 0.002  & 4 & 4 & 0 & 1 & 0.004  & 4 & 4 & 0 & 1 & 0.003  \\ \hline
        (2,2) & 2 & 2 & 0 & 1 & 0.002  & 2 & 2 & 0 & 1 & 0.004  & 2 & 2 & 0 & 1 & 0.002  \\ \hline
        (2,3) & 4 & 3 & 5 & 4 & 0.002  & 4 & 3 & 5 & 4 & 0.004  & 4 & 3 & 5 & 4 & 0.061  \\ \hline
        (2,4) & 6 & 4 & 10 & 4 & 0.002  & 6 & 4 & 10 & 4 & 0.007  & 6 & 4 & 10 & 4 & 0.091  \\ \hline
        (2,5) & 8 & 6 & 10 & 4 & 0.004  & 8 & 6 & 10 & 4 & 0.032  & 8 & 6 & 10 & 4 & 0.142  \\ \hline
        (3,2) & 3 & 3 & 0 & 1 & 0.002  & 3 & 3 & 0 & 1 & 0.004  & 3 & 3 & 0 & 1 & 0.002  \\ \hline
        {\bf (3,3)} & 7 & 3 & 27 & 6 & 0.003  & {\bf 6} & {\bf 4} & 17 & 6 & 0.012  &  {\bf 6} & {\bf 4} &  {\bf 10} & {\bf 4} & 0.341  \\ \hline
        (3,4) & 9 & 4 & 32 & 6 & 0.004  & 9 & 4 & 32 & 6 & 0.013  & 9 & 4 & 32 & 6 & 0.301  \\ \hline
        (3,5) & 12 & 7 & 32 & 6 & 0.006  & 12 & 7 & 32 & 6 & 0.150  & 12 & 7 & 32 & 6 & 4.372  \\ \hline
        (4,2) & 4 & 4 & 0 & 1 & 0.005  & 4 & 4 & 0 & 1 & 0.007  & 4 & 4 & 0 & 1 & 0.003  \\ \hline
        {\bf (4,3)} & 9 & 4 & 57 & 12 & 0.006  & {\bf 8} & {\bf 5} & 47 & 12 & 0.016  & {\bf 8} & {\bf 5} & {\bf 15} & {\bf 4} & 2.812  \\ \hline
        (4,4) & 12 & 4 & 68 & 6 & 0.005  & 12 & 4 & 68 & 6 & 0.438  & 12 & 4 & 68 & 6 & 1.562  \\ \hline
        (4,5) & 16 & 7 & 59 & 6 & 0.007  & 16 & 7 & 59 & 6 & 113.641  & 16 & 7 & 59 & 6 & {\bf 87.922}  \\ \hline
        (5,2) & 5 & 5 & 0 & 1 & 0.012  & 5 & 5 & 0 & 1 & 0.017  & 5 & 5 & 0 & 1 & 0.006  \\ \hline
        {\bf (5,3)} & 11 & 5 & 129 & 21 & 0.014  & {\bf 10} & {\bf 6} & 119 & 21 & 0.069  &  {\bf 10} & {\bf 6} & {\bf 20} & {\bf 4} & 29.482  \\ \hline
        (5,4) & 15 & 5 & 128 & 12 & 0.010  & 15 & 5 & 128 & 12 & {\bf Timeout} & 15 & 5 & 128 & 12 & {\bf 141.412}  \\ \hline
        (5,5) & 20 & 8 & 113 & 12 & 0.018  & 20 & 8 & 113 & 12 & 12539.946  & 20 & 8 & 113 & 12 & {\bf 8849.711}  \\ \hline
        (6,2) & 6 & 6 & 0 & 1 & 0.041  & 6 & 6 & 0 & 1 & 0.051  & 6 & 6 & 0 & 1 & 0.009  \\ \hline
        {\bf (6,3)} & 13 & 6 & 297 & 38 & 0.044  & {\bf 12} & {\bf 7} & 287 & 38 & 4.394  & {\bf 12} & {\bf 7} & {\bf 25} & {\bf 4} & 467.626  \\ \hline
        (6,4) & 18 & 6 & 272 & 21 & 0.038  & 18 & 6 & 272 & 21 & {\bf Timeout} & 18 & 6 & 272 & 21 & {\bf 14878.606}  \\ \hline
        {\bf (6,5)} & 25 & 8 & 244 & 21 & 0.035  & 25 & 8 & 244 & 21 & {\bf Timeout} & {\bf 24} & {\bf 8} & {\bf 136} & {\bf 6} & {\bf Timeout} \\ \hline
    \end{tabular}}
\end{table}



\subsection{Comparison between \tool, AGEFA$_{\tt g}$ and AGEFA$_{\tt bf}$}
We compare \tool with the best setting over AGEFA$_{\tt g}$ and AGEFA$_{\tt bf}$.
%
The results are reported in Table~\ref{tab:code-circuit},
including the number of individual inputs (\#In),
parity size $r$, number of gates (Area),
length of the longest path (Len) in the circuit,
and execution time in seconds. 

First, it is unsurprising that AGEFA$_{\tt g}$ is the most efficient one, because both \tool and AGEFA$_{\tt bf}$ are built upon AGEFA$_{\tt g}$.
We also observe that some code circuits produced by AGEFA$_{\tt g}$, are already (locally) optimal, i.e., the non-boldfaced pairs of $(k,d)$. Thus, the three tools give the same code circuits for these cases. 
But, their local optimality still should be proved by our tool \tool. 
 
For the other cases, both AGEFA$_{\tt g}$ and \tool can 
reduce individual inputs
and AGEFA$_{\tt g}$ can also produce code circuits
that are proved of optimal by \tool. 
However, AGEFA$_{\tt g}$ always has to increase the parity size $r$.
In contrast, \tool can reduce individual inputs
without increasing the parity size, e.g., for $(k,d)=(6,5)$,
on which in 24 hours, \tool finds a better code circuit than AGEFA$_{\tt g}$ and AGEFA$_{\tt bf}$.
\tool is also significantly more efficient than AGEFA$_{\tt bf}$ for hard instances, as 
highlighted in boldface. Moreover,
\tool produces smaller code circuit with lower latency,
attributed to fewer individual inputs, smaller parity size
and more importantly, balanced inputs across coordinate
functions.

We also compare the cryptographic circuits 
generated via an automated tool for error correction using the block cipher PRESENT-80~\cite{BogdanovKLPPRSV07}, 
where both the tool and original PRESENT-80 circuit are provided by~\cite{MullerM24}, while code circuits are provided by \tool and AGEFA$_{\tt g}$ and AGEFA$_{\tt bf}$, respectively. 

The results are reported in Table~\ref{tab:finalcircuits}, including 
the number of individual inputs (\#In) and parity size ($r$)
of the code circuit, the number of gates (Area) and 
length of the longest path (Len) of the final cryptographic circuit before and after applying
the Yosys's {\tt opt} macro command which runs a number of optimization passes. 

Recall that for the cases $(k,d)\in \{(3,5),(4,5),(5,5)\}$,
the code circuits produced by AGEFA$_{\tt g}$ are already optimal,
thus the corresponding cryptographic circuits are the same for three tools. For the other cases, we observe that \tool can significantly reduce the circuit size and latency no matter whether the cryptographic circuit is optimized by Yosys or not except for $(k,d)=(3,3)$.
For $(k,d)=(3,3)$,  while the code circuit produced by \tool
is smaller than that produced by AGEFA$_{\tt bf}$ with the same number of individual inputs and parity size, but their code circuits are different. It indicates that a smaller code circuit does not necessarily results in a smaller cryptographic circuit when some inputs are used multiple times, because their input signals must be individually generated  to satisfy the independence requirement. 
To resolve this issue, the sizes of the sub-circuits generating their input signals should be incorporated into the \synprob problem.
We leave it as interesting future work.
We finally remark that optimizing the entire cryptographic circuit 
may result in vulnerable circuits~\cite{AghaieMRSSS20}, thus 
their security should be proved by dedicated
formal verification tools~\cite{TGCSW23,TanYSCW24}.

\begin{table}[t]
    \centering\setlength{\tabcolsep}{2.5pt} 
    \renewcommand{\arraystretch}{1.1}
    \caption{Results of AGEFA$_{\tt g}$, AGEFA$_{\tt bf}$, and
    \tool for final cryptographic circuits.}
    \vspace{-2mm}
    \label{tab:finalcircuits}
  \scalebox{0.9}{  
  \begin{tabular}{|c||r|r|r|r|r|r||r|r|r|r|r|r||r|r|r|r|r|r|}
    \hline
    \multirow{3}*{($k,d$)} &  
    \multicolumn{6}{c||}{AGEFA$_{\tt g}$} & 
    \multicolumn{6}{c||}{AGEFA$_{\tt bf}$} &
    \multicolumn{6}{c|}{\tool} \\ \cline{2-19}
    & \multirow{2}*{\#In}  & \multirow{2}*{r}& \multicolumn{2}{c|}{Original} & \multicolumn{2}{c||}{Yosys-Opt}& \multirow{2}*{\#In}  & \multirow{2}*{r}& \multicolumn{2}{c|}{Original} & \multicolumn{2}{c||}{Yosys-Opt}& \multirow{2}*{\#In}  & \multirow{2}*{r}& \multicolumn{2}{c|}{Original} & \multicolumn{2}{c|}{Yosys-Opt} \\ \cline{4-7} \cline{10-13} \cline{16-19}
    & & & Area & Len  &  Area & Len 
    & & & Area & Len  &  Area & Len 
    & & & Area & Len  &  Area & Len  \\ \hline
    (3,3) & 7 & 3 & 38035  & 35  & 5536  & 35  & 
            6 & 4 & \textbf{20214}  & 29  & 4097  & 29  & 
            6 & 4 & 24439 & \textbf{27} & \textbf{4001} & \textbf{27} \\ \hline
   (3,5) & 12 & 7 & 131547  & 37  & 11012  & 37  & 
           12 & 7 & 131547 & 37 & 11012 & 37 & 
           12 & 7 & 131547 & 37 & 11012 & 37 \\ \hline
    (4,3) & 9 & 4 & 52984  & 50  & 7205  & 50  & 
            8 & 5 & 27106  & 38  & 4938  & 38  & 
            8 & 5 & \textbf{24827} & \textbf{27} & \textbf{4218} & \textbf{27} \\ \hline
    (4,5) & 16 & 7 & 256649  & 43  & 15932  & 43  & 
            16 & 7 & 256649 & 43 & 15932 & 43 & 
            16 & 7 & 256649 & 43 & 15932 & 43 \\ \hline
    (5,3) & 11 & 5 & 85905  & 78  & 10033  & 78  & 
            10 & 6 & 41483  & 56  & 6331  & 56  & 
            10 & 6 & \textbf{34847} & \textbf{28} & \textbf{4525} & \textbf{28} \\ \hline
        (5,5) & 20 & 8 & 464042  & 64  & 20381  & 64  & 
                20 & 8 & 464042  & 64  & 20381  & 64  & 
                20 & 8 & 464042  & 64  & 20381  & 64  \\ \hline
        (6,3) & 13 & 6 & 139530  & 130  & 15127  & 130  & 
                12 & 7 & 68804  & 90  & 8864  & 90  & 
                12 & 7 & \textbf{65133} & \textbf{33} & \textbf{5744} & \textbf{33} \\ \hline
        (6,5) & 25 & 8 & 545928 & 81 & 22989 & 81 & 
                25 & 8 & 545928 & 81 & 22989 & 81 & 
                \textbf{24} & 8 & \textbf{486484} & \textbf{47} & \textbf{20517} & \textbf{47} \\ \hline
    \end{tabular}}
\end{table}

\section{Related Work}\label{sec:relatedwork}

Fault injection attacks 
are typically mitigated by incorporating redundancy,
including  duplication redundancy that recomputes the output multiple times in parallel or consecutively (e.g.,~\cite{BarElCNTW06,Malkin2006ACC,guo2015security,KarriWMK02,DaemenDEGMP20} and information redundancy from coding theory for which both linear code (e.g.,~\cite{KarriKG03,BertoniBKMP03,OcheretnijKKG05,AzziBCV17,AghaieMRSSS20,ShahmirzadiR020,RasoolzadehS021,BreierKHL20})
and non-linear code (e.g.,~\cite{AkdemirWKS12,KulikowskiWK08})
have been adopted. 
In this work, we consider linear code due to their efficiency, effectiveness and popularity.

Since the design and implementation of optimal countermeasures are intricate and error-prone, 
various approaches were proposed to verify countermeasures or detect flaws
(e.g.,~\cite{nahiyan2016avfsm,saha2018expfault,arribas2020cryptographic,srivastava2020solomon,wang2021sofi,nasahl2022synfi,BDFGZ14,WangMGF18,RichterBrockmann2021FIVERR,TGCSW23,TanYSCW24,GaoSC24,TollecACHJ23}) or repair flaws (e.g.,~\cite{EldibWW16,wang2021sofi,RoyRHB20,roy2022avatar}).  
Prior repair approaches 
 either do not provide security guarantees (e.g.,~\cite{wang2021sofi,RoyRHB20}) or 
are limited to faults induced only by clock glitches~\cite{EldibWW16}). In contrast, our approach automatically synthesizes correct and secure code
circuits for the consolidated fault model~\cite{RichterBrockmann2023RevisitingFA,MullerM24}.

Logic optimization has been extensively studied for optimizing circuits (cf.~\cite{Tu2024,LaMeres2024,Jantsch2023,HungSYYP06}.
It typically transforms circuits
in logical and/or graphical representations (e.g., sum-of-product, product-of-sum, truth table, and directed acyclic graph) and then optimizes by applying transformations (e.g., Boolean algebra rules, Espresso algorithm and Quine-McCluskey algorithm).
Following~\cite{MullerM24},
we use the Quine-McCluskey algorithm to minimize sum-of-products.
However, optimizing code circuits solely is not sufficient in our setting due to domain-specific requirements, while optimizing the entire
cryptographic circuit 
may result in vulnerable circuits~\cite{AghaieMRSSS20}.

\section{Conclusion}\label{sec:conclusion}
We have formulated \synprob, the optimal code circuit synthesis problem of linear codes
for error detection/correction
and proposed a novel algorithm \tool to solve the  problem, prioritizing the minimization
of individual inputs. 
Our algorithm provides not only correct-by-construction and secure-by-construction, but also local optimality guarantees distinguishing from the state-of-the-art tool AGEFA.  
We have implemented our approach in a tool
 and showcased its performance 
for generating optimal code circuits. 
Particularly, the circuits generated by \tool are  better than those generated
by AGEFA.

In future, we plan to improve the scalability and efficiency 
so that it can feature global optimality and 
broadens its applications that go beyond cryptographic circuits, e.g., reliable data transmission for resource-limited devices. 

\begin{acks}
 This study was partially supported by the National Cryptologic Science Foundation of China (Grant No. 2025NCSF01012).
\end{acks}

\bibliographystyle{ACM-Reference-Format}

\begin{thebibliography}{63}


\ifx \showCODEN    \undefined \def \showCODEN     #1{\unskip}     \fi
\ifx \showISBNx    \undefined \def \showISBNx     #1{\unskip}     \fi
\ifx \showISBNxiii \undefined \def \showISBNxiii  #1{\unskip}     \fi
\ifx \showISSN     \undefined \def \showISSN      #1{\unskip}     \fi
\ifx \showLCCN     \undefined \def \showLCCN      #1{\unskip}     \fi
\ifx \shownote     \undefined \def \shownote      #1{#1}          \fi
\ifx \showarticletitle \undefined \def \showarticletitle #1{#1}   \fi
\ifx \showURL      \undefined \def \showURL       {\relax}        \fi
\providecommand\bibfield[2]{#2}
\providecommand\bibinfo[2]{#2}
\providecommand\natexlab[1]{#1}
\providecommand\showeprint[2][]{arXiv:#2}

\bibitem[Aghaie et~al\mbox{.}(2020)]%
        {AghaieMRSSS20}
\bibfield{author}{\bibinfo{person}{Anita Aghaie}, \bibinfo{person}{Amir
  Moradi}, \bibinfo{person}{Shahram Rasoolzadeh}, \bibinfo{person}{Aein~Rezaei
  Shahmirzadi}, \bibinfo{person}{Falk Schellenberg}, {and}
  \bibinfo{person}{Tobias Schneider}.} \bibinfo{year}{2020}\natexlab{}.
\newblock \showarticletitle{Impeccable Circuits}.
\newblock \bibinfo{journal}{\emph{{IEEE} Trans. Computers}}
  \bibinfo{volume}{69}, \bibinfo{number}{3} (\bibinfo{year}{2020}),
  \bibinfo{pages}{361--376}.
\newblock
\href{https://doi.org/10.1109/TC.2019.2948617}{doi:\nolinkurl{10.1109/TC.2019.2948617}}


\bibitem[Agoyan et~al\mbox{.}(2010)]%
        {AgoyanDNRT10}
\bibfield{author}{\bibinfo{person}{Michel Agoyan}, \bibinfo{person}{Jean{-}Max
  Dutertre}, \bibinfo{person}{David Naccache}, \bibinfo{person}{Bruno
  Robisson}, {and} \bibinfo{person}{Assia Tria}.}
  \bibinfo{year}{2010}\natexlab{}.
\newblock \showarticletitle{When Clocks Fail: On Critical Paths and Clock
  Faults}. In \bibinfo{booktitle}{\emph{Proceedings of the 9th {IFIP} {WG}
  8.8/11.2 International Conference ({CARDIS})}}. \bibinfo{pages}{182--193}.
\newblock


\bibitem[Akdemir et~al\mbox{.}(2012)]%
        {AkdemirWKS12}
\bibfield{author}{\bibinfo{person}{Kahraman~D. Akdemir}, \bibinfo{person}{Zhen
  Wang}, \bibinfo{person}{Mark~G. Karpovsky}, {and} \bibinfo{person}{Berk
  Sunar}.} \bibinfo{year}{2012}\natexlab{}.
\newblock \showarticletitle{Design of Cryptographic Devices Resilient to Fault
  Injection Attacks Using Nonlinear Robust Codes}.
\newblock In \bibinfo{booktitle}{\emph{Fault Analysis in Cryptography}},
  \bibfield{editor}{\bibinfo{person}{Marc Joye} {and} \bibinfo{person}{Michael
  Tunstall}} (Eds.). \bibinfo{publisher}{Springer}, \bibinfo{pages}{171--199}.
\newblock
\href{https://doi.org/10.1007/978-3-642-29656-7\_11}{doi:\nolinkurl{10.1007/978-3-642-29656-7\_11}}


\bibitem[Arribas et~al\mbox{.}(2020)]%
        {arribas2020cryptographic}
\bibfield{author}{\bibinfo{person}{Victor Arribas}, \bibinfo{person}{Felix
  Wegener}, \bibinfo{person}{Amir Moradi}, {and} \bibinfo{person}{Svetla
  Nikova}.} \bibinfo{year}{2020}\natexlab{}.
\newblock \showarticletitle{Cryptographic fault diagnosis using VerFI}. In
  \bibinfo{booktitle}{\emph{Proceedings of the IEEE International Symposium on
  Hardware Oriented Security and Trust ({HOST})}}. \bibinfo{pages}{229--240}.
\newblock


\bibitem[Atzori et~al\mbox{.}(2010)]%
        {AtzoriIM10}
\bibfield{author}{\bibinfo{person}{Luigi Atzori}, \bibinfo{person}{Antonio
  Iera}, {and} \bibinfo{person}{Giacomo Morabito}.}
  \bibinfo{year}{2010}\natexlab{}.
\newblock \showarticletitle{The Internet of Things: {A} survey}.
\newblock \bibinfo{journal}{\emph{Comput. Networks}} \bibinfo{volume}{54},
  \bibinfo{number}{15} (\bibinfo{year}{2010}), \bibinfo{pages}{2787--2805}.
\newblock


\bibitem[Azzi et~al\mbox{.}(2017)]%
        {AzziBCV17}
\bibfield{author}{\bibinfo{person}{Sabine Azzi}, \bibinfo{person}{Bruno
  Barras}, \bibinfo{person}{Maria Christofi}, {and} \bibinfo{person}{David
  Vigilant}.} \bibinfo{year}{2017}\natexlab{}.
\newblock \showarticletitle{Using linear codes as a fault countermeasure for
  nonlinear operations: application to {AES} and formal verification}.
\newblock \bibinfo{journal}{\emph{J. Cryptogr. Eng.}} \bibinfo{volume}{7},
  \bibinfo{number}{1} (\bibinfo{year}{2017}), \bibinfo{pages}{75--85}.
\newblock
\href{https://doi.org/10.1007/S13389-016-0138-1}{doi:\nolinkurl{10.1007/S13389-016-0138-1}}


\bibitem[Bar{-}El et~al\mbox{.}(2006)]%
        {BarElCNTW06}
\bibfield{author}{\bibinfo{person}{Hagai Bar{-}El}, \bibinfo{person}{Hamid
  Choukri}, \bibinfo{person}{David Naccache}, \bibinfo{person}{Michael
  Tunstall}, {and} \bibinfo{person}{Claire Whelan}.}
  \bibinfo{year}{2006}\natexlab{}.
\newblock \showarticletitle{The Sorcerer's Apprentice Guide to Fault Attacks}.
\newblock \bibinfo{journal}{\emph{Proc. {IEEE}}} \bibinfo{volume}{94},
  \bibinfo{number}{2} (\bibinfo{year}{2006}), \bibinfo{pages}{370--382}.
\newblock
\href{https://doi.org/10.1109/JPROC.2005.862424}{doi:\nolinkurl{10.1109/JPROC.2005.862424}}


\bibitem[Barbosa et~al\mbox{.}(2022)]%
        {BarbosaBBKLMMMN22}
\bibfield{author}{\bibinfo{person}{Haniel Barbosa}, \bibinfo{person}{Clark~W.
  Barrett}, \bibinfo{person}{Martin Brain}, \bibinfo{person}{Gereon Kremer},
  \bibinfo{person}{Hanna Lachnitt}, \bibinfo{person}{Makai Mann},
  \bibinfo{person}{Abdalrhman Mohamed}, \bibinfo{person}{Mudathir Mohamed},
  \bibinfo{person}{Aina Niemetz}, \bibinfo{person}{Andres N{\"{o}}tzli},
  \bibinfo{person}{Alex Ozdemir}, \bibinfo{person}{Mathias Preiner},
  \bibinfo{person}{Andrew Reynolds}, \bibinfo{person}{Ying Sheng},
  \bibinfo{person}{Cesare Tinelli}, {and} \bibinfo{person}{Yoni Zohar}.}
  \bibinfo{year}{2022}\natexlab{}.
\newblock \showarticletitle{cvc5: {A} Versatile and Industrial-Strength {SMT}
  Solver}.
  \bibinfo{howpublished}{\url{https://github.com/abdoo8080/cvc5/tree/weights}}.
  In \bibinfo{booktitle}{\emph{Proceedings of the 28th International Conference
  on Tools and Algorithms for the Construction and Analysis of Systems
  ({TACAS})}}, Vol.~\bibinfo{volume}{13243}. \bibinfo{publisher}{Springer},
  \bibinfo{pages}{415--442}.
\newblock
\href{https://doi.org/10.1007/978-3-030-99524-9\_24}{doi:\nolinkurl{10.1007/978-3-030-99524-9\_24}}


\bibitem[Barthe et~al\mbox{.}(2014)]%
        {BDFGZ14}
\bibfield{author}{\bibinfo{person}{Gilles Barthe},
  \bibinfo{person}{Fran{\c{c}}ois Dupressoir}, \bibinfo{person}{Pierre{-}Alain
  Fouque}, \bibinfo{person}{Benjamin Gr{\'{e}}goire}, {and}
  \bibinfo{person}{Jean{-}Christophe Zapalowicz}.}
  \bibinfo{year}{2014}\natexlab{}.
\newblock \showarticletitle{Synthesis of Fault Attacks on Cryptographic
  Implementations}. In \bibinfo{booktitle}{\emph{Proceedings of the {ACM}
  {SIGSAC} Conference on Computer and Communications Security}}.
  \bibinfo{pages}{1016--1027}.
\newblock


\bibitem[Bartkewitz et~al\mbox{.}(2022)]%
        {BartkewitzBMMS22}
\bibfield{author}{\bibinfo{person}{Timo Bartkewitz}, \bibinfo{person}{Sven
  Bettendorf}, \bibinfo{person}{Thorben Moos}, \bibinfo{person}{Amir Moradi},
  {and} \bibinfo{person}{Falk Schellenberg}.} \bibinfo{year}{2022}\natexlab{}.
\newblock \showarticletitle{Beware of Insufficient Redundancy An Experimental
  Evaluation of Code-based {FI} Countermeasures}.
\newblock \bibinfo{journal}{\emph{{IACR} Trans. Cryptogr. Hardw. Embed. Syst.}}
  \bibinfo{volume}{2022}, \bibinfo{number}{3} (\bibinfo{year}{2022}),
  \bibinfo{pages}{438--462}.
\newblock
\href{https://doi.org/10.46586/TCHES.V2022.I3.438-462}{doi:\nolinkurl{10.46586/TCHES.V2022.I3.438-462}}


\bibitem[Bertoni et~al\mbox{.}(2003)]%
        {BertoniBKMP03}
\bibfield{author}{\bibinfo{person}{Guido Bertoni}, \bibinfo{person}{Luca
  Breveglieri}, \bibinfo{person}{Israel Koren}, \bibinfo{person}{Paolo
  Maistri}, {and} \bibinfo{person}{Vincenzo Piuri}.}
  \bibinfo{year}{2003}\natexlab{}.
\newblock \showarticletitle{Error Analysis and Detection Procedures for a
  Hardware Implementation of the Advanced Encryption Standard}.
\newblock \bibinfo{journal}{\emph{{IEEE} Trans. Computers}}
  \bibinfo{volume}{52}, \bibinfo{number}{4} (\bibinfo{year}{2003}),
  \bibinfo{pages}{492--505}.
\newblock
\href{https://doi.org/10.1109/TC.2003.1190590}{doi:\nolinkurl{10.1109/TC.2003.1190590}}


\bibitem[Biham and Shamir(1997)]%
        {biham1997differential}
\bibfield{author}{\bibinfo{person}{Eli Biham} {and} \bibinfo{person}{Adi
  Shamir}.} \bibinfo{year}{1997}\natexlab{}.
\newblock \showarticletitle{Differential fault analysis of secret key
  cryptosystems}. In \bibinfo{booktitle}{\emph{Proceedings of the 17th Annual
  International Cryptology Conference Santa Barbara {(CRYPTO)}}}.
  \bibinfo{pages}{513--525}.
\newblock


\bibitem[Blahut(2003)]%
        {Blahut2003}
\bibfield{author}{\bibinfo{person}{Richard~E. Blahut}.}
  \bibinfo{year}{2003}\natexlab{}.
\newblock \bibinfo{booktitle}{\emph{Algebraic Codes for Data Transmission}}.
\newblock \bibinfo{publisher}{Cambridge University Press}.
\newblock


\bibitem[Bogdanov et~al\mbox{.}(2007)]%
        {BogdanovKLPPRSV07}
\bibfield{author}{\bibinfo{person}{Andrey Bogdanov}, \bibinfo{person}{Lars~R.
  Knudsen}, \bibinfo{person}{Gregor Leander}, \bibinfo{person}{Christof Paar},
  \bibinfo{person}{Axel Poschmann}, \bibinfo{person}{Matthew J.~B. Robshaw},
  \bibinfo{person}{Yannick Seurin}, {and} \bibinfo{person}{C. Vikkelsoe}.}
  \bibinfo{year}{2007}\natexlab{}.
\newblock \showarticletitle{{PRESENT:} An Ultra-Lightweight Block Cipher}. In
  \bibinfo{booktitle}{\emph{Proceedings of the 9th International Workshop on
  Cryptographic Hardware and Embedded Systems ({CHES}}},
  \bibfield{editor}{\bibinfo{person}{Pascal Paillier} {and}
  \bibinfo{person}{Ingrid Verbauwhede}} (Eds.), Vol.~\bibinfo{volume}{4727}.
  \bibinfo{pages}{450--466}.
\newblock
\href{https://doi.org/10.1007/978-3-540-74735-2\_31}{doi:\nolinkurl{10.1007/978-3-540-74735-2\_31}}


\bibitem[Boneh et~al\mbox{.}(1997)]%
        {BonehDL97}
\bibfield{author}{\bibinfo{person}{Dan Boneh}, \bibinfo{person}{Richard~A.
  DeMillo}, {and} \bibinfo{person}{Richard~J. Lipton}.}
  \bibinfo{year}{1997}\natexlab{}.
\newblock \showarticletitle{On the Importance of Checking Cryptographic
  Protocols for Faults (Extended Abstract)}. In
  \bibinfo{booktitle}{\emph{Proceeding of the International Conference on the
  Theory and Application of Cryptographic Techniques ({EUROCRYPT})}}
  \emph{(\bibinfo{series}{Lecture Notes in Computer Science},
  Vol.~\bibinfo{volume}{1233})}, \bibfield{editor}{\bibinfo{person}{Walter
  Fumy}} (Ed.). \bibinfo{publisher}{Springer}, \bibinfo{pages}{37--51}.
\newblock
\href{https://doi.org/10.1007/3-540-69053-0\_4}{doi:\nolinkurl{10.1007/3-540-69053-0\_4}}


\bibitem[Breier et~al\mbox{.}(2020)]%
        {BreierKHL20}
\bibfield{author}{\bibinfo{person}{Jakub Breier}, \bibinfo{person}{Mustafa
  Khairallah}, \bibinfo{person}{Xiaolu Hou}, {and} \bibinfo{person}{Yang Liu}.}
  \bibinfo{year}{2020}\natexlab{}.
\newblock \showarticletitle{A Countermeasure Against Statistical Ineffective
  Fault Analysis}.
\newblock \bibinfo{journal}{\emph{{IEEE} Trans. Circuits Syst.}}
  \bibinfo{volume}{67-II}, \bibinfo{number}{12} (\bibinfo{year}{2020}),
  \bibinfo{pages}{3322--3326}.
\newblock
\href{https://doi.org/10.1109/TCSII.2020.2989184}{doi:\nolinkurl{10.1109/TCSII.2020.2989184}}


\bibitem[Clarke et~al\mbox{.}(1998)]%
        {ClarkeEJS98}
\bibfield{author}{\bibinfo{person}{Edmund~M. Clarke}, \bibinfo{person}{E.~Allen
  Emerson}, \bibinfo{person}{Somesh Jha}, {and} \bibinfo{person}{A.~Prasad
  Sistla}.} \bibinfo{year}{1998}\natexlab{}.
\newblock \showarticletitle{Symmetry Reductions in Model Checking}. In
  \bibinfo{booktitle}{\emph{Proceedings of the 10th International Conference on
  Computer Aided Verification ({CAV})}},
  \bibfield{editor}{\bibinfo{person}{Alan~J. Hu} {and}
  \bibinfo{person}{Moshe~Y. Vardi}} (Eds.), Vol.~\bibinfo{volume}{1427}.
  \bibinfo{pages}{147--158}.
\newblock
\href{https://doi.org/10.1007/BFB0028741}{doi:\nolinkurl{10.1007/BFB0028741}}


\bibitem[Clavier(2007)]%
        {clavier2007secret}
\bibfield{author}{\bibinfo{person}{Christophe Clavier}.}
  \bibinfo{year}{2007}\natexlab{}.
\newblock \showarticletitle{Secret external encodings do not prevent transient
  fault analysis}. In \bibinfo{booktitle}{\emph{Proceedings of the 9th
  International Workshop on Cryptographic Hardware and Embedded Systems
  {(CHES)}}}. \bibinfo{pages}{181--194}.
\newblock


\bibitem[Daemen et~al\mbox{.}(2020)]%
        {DaemenDEGMP20}
\bibfield{author}{\bibinfo{person}{Joan Daemen}, \bibinfo{person}{Christoph
  Dobraunig}, \bibinfo{person}{Maria Eichlseder}, \bibinfo{person}{Hannes
  Gro{\ss}}, \bibinfo{person}{Florian Mendel}, {and} \bibinfo{person}{Robert
  Primas}.} \bibinfo{year}{2020}\natexlab{}.
\newblock \showarticletitle{Protecting against Statistical Ineffective Fault
  Attacks}.
\newblock \bibinfo{journal}{\emph{{IACR} Trans. Cryptogr. Hardw. Embed. Syst.}}
  \bibinfo{volume}{2020}, \bibinfo{number}{3} (\bibinfo{year}{2020}),
  \bibinfo{pages}{508--543}.
\newblock
\href{https://doi.org/10.13154/TCHES.V2020.I3.508-543}{doi:\nolinkurl{10.13154/TCHES.V2020.I3.508-543}}


\bibitem[Dehbaoui et~al\mbox{.}(2012)]%
        {DehbaouiDRT12}
\bibfield{author}{\bibinfo{person}{Amine Dehbaoui}, \bibinfo{person}{Jean{-}Max
  Dutertre}, \bibinfo{person}{Bruno Robisson}, {and} \bibinfo{person}{Assia
  Tria}.} \bibinfo{year}{2012}\natexlab{}.
\newblock \showarticletitle{Electromagnetic Transient Faults Injection on a
  Hardware and a Software Implementations of {AES}}. In
  \bibinfo{booktitle}{\emph{Proceedings of the Workshop on Fault Diagnosis and
  Tolerance in Cryptography ({FDTC})}}. \bibinfo{pages}{7--15}.
\newblock


\bibitem[Dobraunig et~al\mbox{.}(2018)]%
        {dobraunig2018statistical}
\bibfield{author}{\bibinfo{person}{Christoph Dobraunig}, \bibinfo{person}{Maria
  Eichlseder}, \bibinfo{person}{Hannes Gro{\ss}}, \bibinfo{person}{Stefan
  Mangard}, \bibinfo{person}{Florian Mendel}, {and} \bibinfo{person}{Robert
  Primas}.} \bibinfo{year}{2018}\natexlab{}.
\newblock \showarticletitle{Statistical ineffective fault attacks on masked
  {AES} with fault countermeasures}. In \bibinfo{booktitle}{\emph{Proceedings
  of the 24th International Conference on the Theory and Application of
  Cryptology and Information Security ((ASIACRYPT))}}.
  \bibinfo{pages}{315--342}.
\newblock


\bibitem[Eldib et~al\mbox{.}(2016)]%
        {EldibWW16}
\bibfield{author}{\bibinfo{person}{Hassan Eldib}, \bibinfo{person}{Meng Wu},
  {and} \bibinfo{person}{Chao Wang}.} \bibinfo{year}{2016}\natexlab{}.
\newblock \showarticletitle{Synthesis of Fault-Attack Countermeasures for
  Cryptographic Circuits}. In \bibinfo{booktitle}{\emph{Proceedings of the 28th
  International Conference on Computer Aided Verification ({CAV})}}.
  \bibinfo{pages}{343--363}.
\newblock


\bibitem[Frandsen and Miltersen(2005)]%
        {FrandsenM05}
\bibfield{author}{\bibinfo{person}{Gudmund~Skovbjerg Frandsen} {and}
  \bibinfo{person}{Peter~Bro Miltersen}.} \bibinfo{year}{2005}\natexlab{}.
\newblock \showarticletitle{Reviewing bounds on the circuit size of the hardest
  functions}.
\newblock \bibinfo{journal}{\emph{Inf. Process. Lett.}} \bibinfo{volume}{95},
  \bibinfo{number}{2} (\bibinfo{year}{2005}), \bibinfo{pages}{354--357}.
\newblock
\href{https://doi.org/10.1016/J.IPL.2005.03.009}{doi:\nolinkurl{10.1016/J.IPL.2005.03.009}}


\bibitem[Fuhr et~al\mbox{.}(2013)]%
        {fuhr2013fault}
\bibfield{author}{\bibinfo{person}{Thomas Fuhr}, \bibinfo{person}{{\'E}liane
  Jaulmes}, \bibinfo{person}{Victor Lomn{\'e}}, {and} \bibinfo{person}{Adrian
  Thillard}.} \bibinfo{year}{2013}\natexlab{}.
\newblock \showarticletitle{Fault attacks on {AES} with faulty ciphertexts
  only}. In \bibinfo{booktitle}{\emph{Proceedings of the Workshop on Fault
  Diagnosis and Tolerance in Cryptography}}. \bibinfo{pages}{108--118}.
\newblock


\bibitem[Gao et~al\mbox{.}(2024)]%
        {GaoSC24}
\bibfield{author}{\bibinfo{person}{Pengfei Gao}, \bibinfo{person}{Fu Song},
  {and} \bibinfo{person}{Taolue Chen}.} \bibinfo{year}{2024}\natexlab{}.
\newblock \showarticletitle{Compositional Verification of First-Order Masking
  Countermeasures against Power Side-Channel Attacks}.
\newblock \bibinfo{journal}{\emph{{ACM} Trans. Softw. Eng. Methodol.}}
  \bibinfo{volume}{33}, \bibinfo{number}{3} (\bibinfo{year}{2024}),
  \bibinfo{pages}{79:1--79:38}.
\newblock


\bibitem[Ghalaty et~al\mbox{.}(2014)]%
        {GhalatyYTS14}
\bibfield{author}{\bibinfo{person}{Nahid~Farhady Ghalaty},
  \bibinfo{person}{Bilgiday Yuce}, \bibinfo{person}{Mostafa Taha}, {and}
  \bibinfo{person}{Patrick Schaumont}.} \bibinfo{year}{2014}\natexlab{}.
\newblock \showarticletitle{Differential Fault Intensity Analysis}. In
  \bibinfo{booktitle}{\emph{Proceeding of the Workshop on Fault Diagnosis and
  Tolerance in Cryptography}}. \bibinfo{pages}{49--58}.
\newblock
\href{https://doi.org/10.1109/FDTC.2014.15}{doi:\nolinkurl{10.1109/FDTC.2014.15}}


\bibitem[Guo et~al\mbox{.}(2015)]%
        {guo2015security}
\bibfield{author}{\bibinfo{person}{Xiaofei Guo}, \bibinfo{person}{Debdeep
  Mukhopadhyay}, \bibinfo{person}{Chenglu Jin}, {and} \bibinfo{person}{Ramesh
  Karri}.} \bibinfo{year}{2015}\natexlab{}.
\newblock \showarticletitle{Security analysis of concurrent error detection
  against differential fault analysis}.
\newblock \bibinfo{journal}{\emph{J. Cryptogr. Eng.}} \bibinfo{volume}{5},
  \bibinfo{number}{3} (\bibinfo{year}{2015}), \bibinfo{pages}{153--169}.
\newblock
\href{https://doi.org/10.1007/S13389-014-0092-8}{doi:\nolinkurl{10.1007/S13389-014-0092-8}}


\bibitem[Hung et~al\mbox{.}(2006)]%
        {HungSYYP06}
\bibfield{author}{\bibinfo{person}{William N.~N. Hung}, \bibinfo{person}{Xiaoyu
  Song}, \bibinfo{person}{Guowu Yang}, \bibinfo{person}{Jin Yang}, {and}
  \bibinfo{person}{Marek~A. Perkowski}.} \bibinfo{year}{2006}\natexlab{}.
\newblock \showarticletitle{Optimal synthesis of multiple output Boolean
  functions using a set of quantum gates by symbolic reachability analysis}.
\newblock \bibinfo{journal}{\emph{{IEEE} Trans. Comput. Aided Des. Integr.
  Circuits Syst.}} \bibinfo{volume}{25}, \bibinfo{number}{9}
  (\bibinfo{year}{2006}), \bibinfo{pages}{1652--1663}.
\newblock
\href{https://doi.org/10.1109/TCAD.2005.858352}{doi:\nolinkurl{10.1109/TCAD.2005.858352}}


\bibitem[Jantsch(2023)]%
        {Jantsch2023}
\bibfield{author}{\bibinfo{person}{Axel Jantsch}.}
  \bibinfo{year}{2023}\natexlab{}.
\newblock \bibinfo{booktitle}{\emph{Logic Optimization}}.
\newblock \bibinfo{publisher}{Springer Nature Switzerland},
  \bibinfo{address}{Cham}, \bibinfo{pages}{179--225}.
\newblock


\bibitem[Karri et~al\mbox{.}(2003)]%
        {KarriKG03}
\bibfield{author}{\bibinfo{person}{Ramesh Karri}, \bibinfo{person}{Grigori
  Kuznetsov}, {and} \bibinfo{person}{Michael G{\"{o}}ssel}.}
  \bibinfo{year}{2003}\natexlab{}.
\newblock \showarticletitle{Parity-Based Concurrent Error Detection of
  Substitution-Permutation Network Block Ciphers}. In
  \bibinfo{booktitle}{\emph{Proceedings of the 5th International Workshop on
  Cryptographic Hardware and Embedded Systems ({CHES})}}.
  \bibinfo{pages}{113--124}.
\newblock


\bibitem[Karri et~al\mbox{.}(2002)]%
        {KarriWMK02}
\bibfield{author}{\bibinfo{person}{Ramesh Karri}, \bibinfo{person}{Kaijie Wu},
  \bibinfo{person}{Piyush Mishra}, {and} \bibinfo{person}{Yongkook Kim}.}
  \bibinfo{year}{2002}\natexlab{}.
\newblock \showarticletitle{Concurrent error detection schemes for fault-based
  side-channel cryptanalysis of symmetric block ciphers}.
\newblock \bibinfo{journal}{\emph{{IEEE} Trans. Comput. Aided Des. Integr.
  Circuits Syst.}} \bibinfo{volume}{21}, \bibinfo{number}{12}
  (\bibinfo{year}{2002}), \bibinfo{pages}{1509--1517}.
\newblock
\href{https://doi.org/10.1109/TCAD.2002.804378}{doi:\nolinkurl{10.1109/TCAD.2002.804378}}


\bibitem[Kulikowski et~al\mbox{.}(2008)]%
        {KulikowskiWK08}
\bibfield{author}{\bibinfo{person}{Konrad~J. Kulikowski}, \bibinfo{person}{Zhen
  Wang}, {and} \bibinfo{person}{Mark~G. Karpovsky}.}
  \bibinfo{year}{2008}\natexlab{}.
\newblock \showarticletitle{Comparative Analysis of Robust Fault Attack
  Resistant Architectures for Public and Private Cryptosystems}. In
  \bibinfo{booktitle}{\emph{Proceedings of the 5th International Workshop on
  Fault Diagnosis and Tolerance in Cryptography ({FDTC}}},
  \bibfield{editor}{\bibinfo{person}{Luca Breveglieri}, \bibinfo{person}{Shay
  Gueron}, \bibinfo{person}{Israel Koren}, \bibinfo{person}{David Naccache},
  {and} \bibinfo{person}{Jean{-}Pierre Seifert}} (Eds.).
  \bibinfo{publisher}{{IEEE} Computer Society}, \bibinfo{pages}{41--50}.
\newblock
\href{https://doi.org/10.1109/FDTC.2008.13}{doi:\nolinkurl{10.1109/FDTC.2008.13}}


\bibitem[LaMeres(2024)]%
        {LaMeres2024}
\bibfield{author}{\bibinfo{person}{Brock~J. LaMeres}.}
  \bibinfo{year}{2024}\natexlab{}.
\newblock \bibinfo{booktitle}{\emph{Combinational Logic Design}}.
\newblock \bibinfo{publisher}{Springer International Publishing},
  \bibinfo{address}{Cham}, \bibinfo{pages}{93--154}.
\newblock
\showISBNx{978-3-031-42547-9}
\href{https://doi.org/10.1007/978-3-031-42547-9_4}{doi:\nolinkurl{10.1007/978-3-031-42547-9_4}}


\bibitem[Li et~al\mbox{.}(2010)]%
        {LiSGFTO10}
\bibfield{author}{\bibinfo{person}{Yang Li}, \bibinfo{person}{Kazuo Sakiyama},
  \bibinfo{person}{Shigeto Gomisawa}, \bibinfo{person}{Toshinori Fukunaga},
  \bibinfo{person}{Junko Takahashi}, {and} \bibinfo{person}{Kazuo Ohta}.}
  \bibinfo{year}{2010}\natexlab{}.
\newblock \showarticletitle{Fault Sensitivity Analysis}. In
  \bibinfo{booktitle}{\emph{Proceedings of the 12th International Workshop on
  Cryptographic Hardware and Embedded Systems ({CHES})}},
  \bibfield{editor}{\bibinfo{person}{Stefan Mangard} {and}
  \bibinfo{person}{Fran{\c{c}}ois{-}Xavier Standaert}} (Eds.),
  Vol.~\bibinfo{volume}{6225}. \bibinfo{publisher}{Springer},
  \bibinfo{pages}{320--334}.
\newblock
\href{https://doi.org/10.1007/978-3-642-15031-9\_22}{doi:\nolinkurl{10.1007/978-3-642-15031-9\_22}}


\bibitem[Lupanov(1958)]%
        {Lupanov58}
\bibfield{author}{\bibinfo{person}{O.B. Lupanov}.}
  \bibinfo{year}{1958}\natexlab{}.
\newblock \showarticletitle{The synthesis of contact circuits}.
\newblock \bibinfo{journal}{\emph{Dokl. Akad. Nauk SSSR}}
  \bibinfo{volume}{119}, \bibinfo{number}{1} (\bibinfo{year}{1958}),
  \bibinfo{pages}{23--26}.
\newblock


\bibitem[MacWilliams and Sloane(1977)]%
        {MacWilliamsS77}
\bibfield{author}{\bibinfo{person}{F.~Jessie MacWilliams} {and}
  \bibinfo{person}{Neil J.~A. Sloane}.} \bibinfo{year}{1977}\natexlab{}.
\newblock \bibinfo{booktitle}{\emph{The Theory of Error-Correcting Codes}}.
\newblock \bibinfo{publisher}{North-Holland Pub. Co.}
\newblock


\bibitem[Malkin et~al\mbox{.}(2006)]%
        {Malkin2006ACC}
\bibfield{author}{\bibinfo{person}{Tal Malkin},
  \bibinfo{person}{Fran{\c{c}}ois{-}Xavier Standaert}, {and}
  \bibinfo{person}{Moti Yung}.} \bibinfo{year}{2006}\natexlab{}.
\newblock \showarticletitle{A Comparative Cost/Security Analysis of Fault
  Attack Countermeasures}. In \bibinfo{booktitle}{\emph{Proceedings of the 3rd
  International Workshop on Fault Diagnosis and Tolerance in Cryptography}}.
  \bibinfo{pages}{159--172}.
\newblock


\bibitem[McCluskey(1956)]%
        {mccluskey1956}
\bibfield{author}{\bibinfo{person}{Edward~J McCluskey}.}
  \bibinfo{year}{1956}\natexlab{}.
\newblock \showarticletitle{Minimization of Boolean functions}.
\newblock \bibinfo{journal}{\emph{The Bell System Technical Journal}}
  \bibinfo{volume}{35}, \bibinfo{number}{6} (\bibinfo{year}{1956}),
  \bibinfo{pages}{1417--1444}.
\newblock


\bibitem[M{\"{u}}ller and Moradi(2024)]%
        {MullerM24}
\bibfield{author}{\bibinfo{person}{Nicolai M{\"{u}}ller} {and}
  \bibinfo{person}{Amir Moradi}.} \bibinfo{year}{2024}\natexlab{}.
\newblock \showarticletitle{Automated Generation of Fault-Resistant Circuits}.
\newblock \bibinfo{journal}{\emph{{IACR} Trans. Cryptogr. Hardw. Embed. Syst.}}
  \bibinfo{volume}{2024}, \bibinfo{number}{3} (\bibinfo{year}{2024}),
  \bibinfo{pages}{136--173}.
\newblock
\href{https://doi.org/10.46586/TCHES.V2024.I3.136-173}{doi:\nolinkurl{10.46586/TCHES.V2024.I3.136-173}}


\bibitem[Nahiyan et~al\mbox{.}(2016)]%
        {nahiyan2016avfsm}
\bibfield{author}{\bibinfo{person}{Adib Nahiyan}, \bibinfo{person}{Kan Xiao},
  \bibinfo{person}{Kun Yang}, \bibinfo{person}{Yeir Jin},
  \bibinfo{person}{Domenic Forte}, {and} \bibinfo{person}{Mark Tehranipoor}.}
  \bibinfo{year}{2016}\natexlab{}.
\newblock \showarticletitle{AVFSM: A framework for identifying and mitigating
  vulnerabilities in FSMs}. In \bibinfo{booktitle}{\emph{Proceedings of the
  53rd Annual Design Automation Conference}}. \bibinfo{pages}{1--6}.
\newblock


\bibitem[Nasahl et~al\mbox{.}(2022)]%
        {nasahl2022synfi}
\bibfield{author}{\bibinfo{person}{Pascal Nasahl}, \bibinfo{person}{Miguel
  Osorio}, \bibinfo{person}{Pirmin Vogel}, \bibinfo{person}{Michael Schaffner},
  \bibinfo{person}{Timothy Trippel}, \bibinfo{person}{Dominic Rizzo}, {and}
  \bibinfo{person}{Stefan Mangard}.} \bibinfo{year}{2022}\natexlab{}.
\newblock \showarticletitle{{SYNFI:} Pre-Silicon Fault Analysis of an
  Open-Source Secure Element}.
\newblock \bibinfo{journal}{\emph{{IACR} Transactions on Cryptographic Hardware
  and Embedded Systems}} \bibinfo{volume}{2022}, \bibinfo{number}{4}
  (\bibinfo{year}{2022}), \bibinfo{pages}{56--87}.
\newblock


\bibitem[Ocheretnij et~al\mbox{.}(2005)]%
        {OcheretnijKKG05}
\bibfield{author}{\bibinfo{person}{Vitalij Ocheretnij}, \bibinfo{person}{G.
  Kouznetsov}, \bibinfo{person}{Ramesh Karri}, {and} \bibinfo{person}{Michael
  G{\"{o}}ssel}.} \bibinfo{year}{2005}\natexlab{}.
\newblock \showarticletitle{On-Line Error Detection and {BIST} for the {AES}
  Encryption Algorithm with Different S-Box Implementations}. In
  \bibinfo{booktitle}{\emph{Proceedings of the {IEEE} International On-Line
  Testing Symposium {(IOLTS})}}. \bibinfo{publisher}{{IEEE} Computer Society},
  \bibinfo{pages}{141--146}.
\newblock
\href{https://doi.org/10.1109/IOLTS.2005.51}{doi:\nolinkurl{10.1109/IOLTS.2005.51}}


\bibitem[Quine(1952)]%
        {quine1952}
\bibfield{author}{\bibinfo{person}{Willard~V Quine}.}
  \bibinfo{year}{1952}\natexlab{}.
\newblock \showarticletitle{The problem of simplifying truth functions}.
\newblock \bibinfo{journal}{\emph{The American mathematical monthly}}
  \bibinfo{volume}{59}, \bibinfo{number}{8} (\bibinfo{year}{1952}),
  \bibinfo{pages}{521--531}.
\newblock


\bibitem[Rasoolzadeh et~al\mbox{.}(2021)]%
        {RasoolzadehS021}
\bibfield{author}{\bibinfo{person}{Shahram Rasoolzadeh},
  \bibinfo{person}{Aein~Rezaei Shahmirzadi}, {and} \bibinfo{person}{Amir
  Moradi}.} \bibinfo{year}{2021}\natexlab{}.
\newblock \showarticletitle{Impeccable Circuits {III}}. In
  \bibinfo{booktitle}{\emph{Proceedings of the {IEEE} International Test
  Conference ({ITC})}}. \bibinfo{publisher}{{IEEE}}, \bibinfo{pages}{163--169}.
\newblock
\href{https://doi.org/10.1109/ITC50571.2021.00024}{doi:\nolinkurl{10.1109/ITC50571.2021.00024}}


\bibitem[Richter-Brockmann et~al\mbox{.}(2023)]%
        {RichterBrockmann2023RevisitingFA}
\bibfield{author}{\bibinfo{person}{Jan Richter-Brockmann},
  \bibinfo{person}{Pascal Sasdrich}, {and} \bibinfo{person}{Tim G{\"u}neysu}.}
  \bibinfo{year}{2023}\natexlab{}.
\newblock \showarticletitle{Revisiting Fault Adversary Models - Hardware Faults
  in Theory and Practice}.
\newblock \bibinfo{journal}{\emph{IEEE Trans. Comput.}}  \bibinfo{volume}{72}
  (\bibinfo{year}{2023}), \bibinfo{pages}{572--585}.
\newblock


\bibitem[Richter-Brockmann et~al\mbox{.}(2021)]%
        {RichterBrockmann2021FIVERR}
\bibfield{author}{\bibinfo{person}{Jan Richter-Brockmann},
  \bibinfo{person}{Aein~Rezaei Shahmirzadi}, \bibinfo{person}{Pascal Sasdrich},
  \bibinfo{person}{Amir Moradi}, {and} \bibinfo{person}{Tim G{\"u}neysu}.}
  \bibinfo{year}{2021}\natexlab{}.
\newblock \showarticletitle{FIVER - Robust Verification of Countermeasures
  against Fault Injections}.
\newblock \bibinfo{journal}{\emph{{IACR} Transactions on Cryptographic Hardware
  and Embedded Systems}}  \bibinfo{volume}{2021} (\bibinfo{year}{2021}),
  \bibinfo{pages}{447--473}.
\newblock


\bibitem[Roy et~al\mbox{.}(2020)]%
        {RoyRHB20}
\bibfield{author}{\bibinfo{person}{Indrani Roy}, \bibinfo{person}{Chester
  Rebeiro}, \bibinfo{person}{Aritra Hazra}, {and} \bibinfo{person}{Swarup
  Bhunia}.} \bibinfo{year}{2020}\natexlab{}.
\newblock \showarticletitle{{SAFARI:} Automatic Synthesis of Fault-Attack
  Resistant Block Cipher Implementations}.
\newblock \bibinfo{journal}{\emph{{IEEE} Trans. Comput. Aided Des. Integr.
  Circuits Syst.}} \bibinfo{volume}{39}, \bibinfo{number}{4}
  (\bibinfo{year}{2020}), \bibinfo{pages}{752--765}.
\newblock


\bibitem[Roy et~al\mbox{.}(2022)]%
        {roy2022avatar}
\bibfield{author}{\bibinfo{person}{Prithwish~Basu Roy},
  \bibinfo{person}{Patanjali Slpsk}, {and} \bibinfo{person}{Chester Rebeiro}.}
  \bibinfo{year}{2022}\natexlab{}.
\newblock \showarticletitle{Avatar: Reinforcing fault attack countermeasures in
  EDA with fault transformations}. In \bibinfo{booktitle}{\emph{2022 27th Asia
  and South Pacific Design Automation Conference (ASP-DAC)}}. IEEE,
  \bibinfo{pages}{417--422}.
\newblock


\bibitem[Saha et~al\mbox{.}(2018)]%
        {saha2018expfault}
\bibfield{author}{\bibinfo{person}{Sayandeep Saha}, \bibinfo{person}{Debdeep
  Mukhopadhyay}, {and} \bibinfo{person}{Pallab Dasgupta}.}
  \bibinfo{year}{2018}\natexlab{}.
\newblock \showarticletitle{ExpFault: An Automated Framework for Exploitable
  Fault Characterization in Block Ciphers}.
\newblock \bibinfo{journal}{\emph{{IACR} Transactions on Cryptographic Hardware
  and Embedded Systems}} \bibinfo{volume}{2018}, \bibinfo{number}{2}
  (\bibinfo{year}{2018}), \bibinfo{pages}{242--276}.
\newblock


\bibitem[Selmane et~al\mbox{.}(2008)]%
        {elmaneGD08}
\bibfield{author}{\bibinfo{person}{Nidhal Selmane}, \bibinfo{person}{Sylvain
  Guilley}, {and} \bibinfo{person}{Jean{-}Luc Danger}.}
  \bibinfo{year}{2008}\natexlab{}.
\newblock \showarticletitle{Practical Setup Time Violation Attacks on {AES}}.
  In \bibinfo{booktitle}{\emph{Proceedings of the 7th European Dependable
  Computing Conference ({EDCC})}}. \bibinfo{pages}{91--96}.
\newblock


\bibitem[Shahmirzadi et~al\mbox{.}(2020)]%
        {ShahmirzadiR020}
\bibfield{author}{\bibinfo{person}{Aein~Rezaei Shahmirzadi},
  \bibinfo{person}{Shahram Rasoolzadeh}, {and} \bibinfo{person}{Amir Moradi}.}
  \bibinfo{year}{2020}\natexlab{}.
\newblock \showarticletitle{Impeccable Circuits {II}}. In
  \bibinfo{booktitle}{\emph{Proceedings of the {ACM/IEEE} Design Automation
  Conference ({DAC})}}. \bibinfo{publisher}{{IEEE}}, \bibinfo{pages}{1--6}.
\newblock
\href{https://doi.org/10.1109/DAC18072.2020.9218615}{doi:\nolinkurl{10.1109/DAC18072.2020.9218615}}


\bibitem[Skorobogatov and Anderson(2003)]%
        {skorobogatov2003optical}
\bibfield{author}{\bibinfo{person}{Sergei~P Skorobogatov} {and}
  \bibinfo{person}{Ross~J Anderson}.} \bibinfo{year}{2003}\natexlab{}.
\newblock \showarticletitle{Optical fault induction attacks}. In
  \bibinfo{booktitle}{\emph{Proceedings of the 4th International Workshop
  Redwood Shores on Cryptographic Hardware and Embedded Systems ({CHES})}}.
  \bibinfo{pages}{2--12}.
\newblock


\bibitem[Srivastava et~al\mbox{.}(2020)]%
        {srivastava2020solomon}
\bibfield{author}{\bibinfo{person}{Milind Srivastava},
  \bibinfo{person}{Patanjali Slpsk}, \bibinfo{person}{Indrani Roy},
  \bibinfo{person}{Chester Rebeiro}, \bibinfo{person}{Aritra Hazra}, {and}
  \bibinfo{person}{Swarup Bhunia}.} \bibinfo{year}{2020}\natexlab{}.
\newblock \showarticletitle{SOLOMON: An automated framework for detecting fault
  attack vulnerabilities in hardware}. In \bibinfo{booktitle}{\emph{Proceedings
  of the Design, Automation \& Test in Europe Conference \& Exhibition
  ({DATE})}}. \bibinfo{pages}{310--313}.
\newblock


\bibitem[Tan et~al\mbox{.}(2023)]%
        {TGCSW23}
\bibfield{author}{\bibinfo{person}{Huiyu Tan}, \bibinfo{person}{Pengfei Gao},
  \bibinfo{person}{Taolue Chen}, \bibinfo{person}{Fu Song}, {and}
  \bibinfo{person}{Zhilin Wu}.} \bibinfo{year}{2023}\natexlab{}.
\newblock \showarticletitle{{SAT}-based Formal Fault-Resistance Verification of
  Cryptographic Circuits}.
\newblock \bibinfo{journal}{\emph{CoRR}}  \bibinfo{volume}{abs/2307.00561}
  (\bibinfo{year}{2023}).
\newblock


\bibitem[Tan et~al\mbox{.}(2024)]%
        {TanYSCW24}
\bibfield{author}{\bibinfo{person}{Huiyu Tan}, \bibinfo{person}{Xi Yang},
  \bibinfo{person}{Fu Song}, \bibinfo{person}{Taolue Chen}, {and}
  \bibinfo{person}{Zhilin Wu}.} \bibinfo{year}{2024}\natexlab{}.
\newblock \showarticletitle{Compositional Verification of Cryptographic
  Circuits Against Fault Injection Attacks}. In
  \bibinfo{booktitle}{\emph{Proceedings of the 26th International Symposium on
  Formal Methods ({FM}), Part {II}}},
  \bibfield{editor}{\bibinfo{person}{Andr{\'{e}} Platzer},
  \bibinfo{person}{Kristin~Yvonne Rozier}, \bibinfo{person}{Matteo Pradella},
  {and} \bibinfo{person}{Matteo Rossi}} (Eds.), Vol.~\bibinfo{volume}{14934}.
  \bibinfo{pages}{189--207}.
\newblock


\bibitem[Tollec et~al\mbox{.}(2023)]%
        {TollecACHJ23}
\bibfield{author}{\bibinfo{person}{Simon Tollec}, \bibinfo{person}{Mihail
  Asavoae}, \bibinfo{person}{Damien Courouss{\'{e}}}, \bibinfo{person}{Karine
  Heydemann}, {and} \bibinfo{person}{Mathieu Jan}.}
  \bibinfo{year}{2023}\natexlab{}.
\newblock \showarticletitle{{\(\mu\)}ARCHIFI: Formal Modeling and Verification
  Strategies for Microarchitectural Fault Injections}. In
  \bibinfo{booktitle}{\emph{Proceedings of the Formal Methods in Computer-Aided
  Design}}, \bibfield{editor}{\bibinfo{person}{Alexander Nadel} {and}
  \bibinfo{person}{Kristin~Yvonne Rozier}} (Eds.). \bibinfo{pages}{101--109}.
\newblock


\bibitem[Tu et~al\mbox{.}(2024)]%
        {Tu2024}
\bibfield{author}{\bibinfo{person}{Kaihui Tu}, \bibinfo{person}{Xifan Tang},
  \bibinfo{person}{Cunxi Yu}, \bibinfo{person}{Lana Josipovi{\'{c}}}, {and}
  \bibinfo{person}{Zhufei Chu}.} \bibinfo{year}{2024}\natexlab{}.
\newblock \bibinfo{booktitle}{\emph{Logic Synthesis}}.
\newblock \bibinfo{publisher}{Springer Nature Singapore},
  \bibinfo{address}{Singapore}, \bibinfo{pages}{135--164}.
\newblock


\bibitem[Tyagi and Sreenath(2021)]%
        {TYAGI202122}
\bibfield{author}{\bibinfo{person}{Amit~Kumar Tyagi} {and} \bibinfo{person}{N.
  Sreenath}.} \bibinfo{year}{2021}\natexlab{}.
\newblock \showarticletitle{Cyber Physical Systems: Analyses, challenges and
  possible solutions}.
\newblock \bibinfo{journal}{\emph{Internet of Things and Cyber-Physical
  Systems}}  \bibinfo{volume}{1} (\bibinfo{year}{2021}),
  \bibinfo{pages}{22--33}.
\newblock
\showISSN{2667-3452}


\bibitem[van Lint(1999)]%
        {Lint88}
\bibfield{author}{\bibinfo{person}{J.H. van Lint}.}
  \bibinfo{year}{1999}\natexlab{}.
\newblock \bibinfo{booktitle}{\emph{Introduction to coding theory}
  (\bibinfo{edition}{3rd rev. and expanded ed.} ed.)}.
\newblock \bibinfo{publisher}{Springer}.
\newblock
\showISBNx{3-540-64133-5}


\bibitem[Wang et~al\mbox{.}(2021)]%
        {wang2021sofi}
\bibfield{author}{\bibinfo{person}{Huanyu Wang}, \bibinfo{person}{Henian Li},
  \bibinfo{person}{Fahim Rahman}, \bibinfo{person}{Mark~M Tehranipoor}, {and}
  \bibinfo{person}{Farimah Farahmandi}.} \bibinfo{year}{2021}\natexlab{}.
\newblock \showarticletitle{{SoFI}: Security property-driven vulnerability
  assessments of {IC}s against fault-injection attacks}.
\newblock \bibinfo{journal}{\emph{IEEE Transactions on Computer-Aided Design of
  Integrated Circuits and Systems}} \bibinfo{volume}{41}, \bibinfo{number}{3}
  (\bibinfo{year}{2021}), \bibinfo{pages}{452--465}.
\newblock


\bibitem[Wang et~al\mbox{.}(2018)]%
        {WangMGF18}
\bibfield{author}{\bibinfo{person}{Peikun Wang}, \bibinfo{person}{Conrad~J.
  Moore}, \bibinfo{person}{Amir~Masoud Gharehbaghi}, {and}
  \bibinfo{person}{Masahiro Fujita}.} \bibinfo{year}{2018}\natexlab{}.
\newblock \showarticletitle{An {ATPG} Method for Double Stuck-At Faults by
  Analyzing Propagation Paths of Single Faults}.
\newblock \bibinfo{journal}{\emph{{IEEE} Trans. Circuits Syst. {I} Regul.
  Pap.}} \bibinfo{volume}{65-I}, \bibinfo{number}{3} (\bibinfo{year}{2018}),
  \bibinfo{pages}{1063--1074}.
\newblock


\bibitem[Wolf({[n.\,d.]})]%
        {Yosys}
\bibfield{author}{\bibinfo{person}{Claire Wolf}.}
  \bibinfo{year}{[n.\,d.]}\natexlab{}.
\newblock \bibinfo{title}{Yosys Open SYnthesis Suite}.
\newblock \bibinfo{howpublished}{\url{https://yosyshq.net/yosys}}.
\newblock


\bibitem[Zussa et~al\mbox{.}(2013)]%
        {zussa2013power}
\bibfield{author}{\bibinfo{person}{Loic Zussa}, \bibinfo{person}{Jean-Max
  Dutertre}, \bibinfo{person}{Jessy Clediere}, {and} \bibinfo{person}{Assia
  Tria}.} \bibinfo{year}{2013}\natexlab{}.
\newblock \showarticletitle{Power supply glitch induced faults on FPGA: An
  in-depth analysis of the injection mechanism}. In
  \bibinfo{booktitle}{\emph{Proceedings of the {IEEE} 19th International
  On-Line Testing Symposium (IOLTS)}}. \bibinfo{pages}{110--115}.
\newblock


\end{thebibliography}

%


\end{document}